\newif\ifcomments     
\newif\ifanonymous    
\newif\ifextended     
\newif\ifsubmission   
\newif\ifpublic       

\commentsfalse        
\anonymousfalse
\extendedtrue

\submissionfalse       
\publictrue            

\ifpublic
\submissionfalse
\commentstrue
\anonymousfalse
\fi

\ifanonymous
\documentclass{llncs}
\else
\documentclass[runningheads]{llncs}
\ifpublic
\fi
\fi

\usepackage{savesym}

\usepackage{bigfoot}
\usepackage{amsmath,amssymb}
\usepackage{ottalt}
\usepackage{xcolor}
\usepackage{graphicx}
\usepackage{verbatim}
\usepackage{bm}
\usepackage[utf8]{inputenc}
\usepackage{natbib}
\usepackage[hyphens]{url}
\usepackage{textcomp}
\usepackage{wasysym}
\usepackage[hidelinks]{hyperref}
 
\savesymbol{Square}
\savesymbol{Circle}
\usepackage[misc,geometry]{ifsym}

\def\orcidID#1{\smash{\href{http://orcid.org/#1}{\protect\raisebox{-1.25pt}{\protect\includegraphics{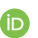}}}}}

\definecolor{light-gray}{gray}{0.90}

\usepackage[barr]{xy}
\let\mto\to
\renewcommand{\to}{\rightarrow}

\ifcomments
\newcommand{\scw}[1]{\textcolor{blue}{{SCW: #1}}}
\newcommand{\rae}[1]{\textcolor{magenta}{{RAE: #1}}}
\newcommand{\hde}[1]{\textcolor{gray}{{HDE: #1}}}
\newcommand{\pc}[1]{\textcolor[rgb]{0.5, 0.0, 0.0}{{PC: #1}}}
\else
\newcommand{\scw}[1]{}
\newcommand{\rae}[1]{}
\newcommand{\hde}[1]{}
\newcommand{\pc}[1]{}
\fi

\ifextended
\newcommand{\auxiliarymaterial}{the appendix}
\newcommand{\auxref}[1]{Appendix~\ref{#1}}
\else
\ifsubmission
\newcommand{\auxiliarymaterial}{the anonymized supplementary material}
\else
\newcommand{\auxiliarymaterial}{the extended version of this paper~\citep{DDC-arXiv}}
\fi
\newcommand{\auxref}[1]{\auxiliarymaterial}
\fi

\usepackage{listingsutf8}
\usepackage{lstparams}
\newcommand{\cd}[1]{\lstinline!#1!}
\newcommand{\DDC}{\textsc{DDC}}
\newcommand{\DDCt}{\DDC$^{\top}$}

\newenvironment{ottdefnblock}[3][]{ \framebox{\mbox{#2}} \quad #3 \\[0pt]}{}

\newcommand{\ottnt}[1]{\mathit{#1}}
\newcommand{\ottmv}[1]{\mathit{#1}}
\newcommand{\ottkw}[1]{\mathbf{#1}}
\newcommand{\ottsym}[1]{#1}

\newcommand{\ottafterlastrule}{\\}

  \renewottcommands[ott]

\DeclareUnicodeCharacter{2200}{\ensuremath{\forall}} 
\DeclareUnicodeCharacter{03A3}{\ensuremath{\Sigma}}  
\DeclareUnicodeCharacter{03B1}{\ensuremath{\alpha}} 
\DeclareUnicodeCharacter{03BB}{\ensuremath{\lambda}} 
\DeclareUnicodeCharacter{2208}{\ensuremath{\in}}
\DeclareUnicodeCharacter{2261}{\ensuremath{\equiv}}
\DeclareUnicodeCharacter{2102}{\ensuremath{\mathbb{C}}}
\DeclareUnicodeCharacter{2115}{\ensuremath{\mathbb{N}}}
\DeclareUnicodeCharacter{03C0}{\ensuremath{\pi}}
\DeclareUnicodeCharacter{2081}{\ensuremath{_1}}
\DeclareUnicodeCharacter{2082}{\ensuremath{_2}}
 
\AtBeginDocument{%
\RestyleFootnote{default}{para}}

\begin{document}

\bibliographystyle{splncs04nat}

\ifextended
\title{A Dependent Dependency Calculus (Extended Version)}
\fi

\ifanonymous
\author{Anonymous authors}
\else
\author{Pritam Choudhury \inst{1}\Letter \and 
       Harley {Eades III}\inst{2}
       \orcidID{0000-0001-8474-5971}\and
       Stephanie Weirich\inst{3}
       \orcidID{0000-0002-6756-9168}}
\institute{University of Pennsylvania, Philadelphia, PA, USA, \email{pritam@seas.upenn.edu}
  \and Augusta University, Augusta, GA, USA
  \and University of Pennsylvania, Philadelphia, PA, USA }
\fi

\authorrunning{Pritam Choudhury, Harley Eades III, and Stephanie Weirich}

\maketitle
\ifextended
\thispagestyle{plain}
\pagestyle{plain}
\fi

\begin{abstract}
Over twenty years ago, Abadi et al. established the Dependency Core Calculus
(DCC) as a general purpose framework for analyzing dependency in typed
programming languages. Since then, dependency analysis has shown many
practical benefits to language design: its results can help users and
compilers enforce security constraints, eliminate dead code, among other applications.
In this work, we present a Dependent
Dependency Calculus (DDC), which extends this general idea to the
setting of a dependently-typed language. We use this calculus to track both
run-time and compile-time irrelevance, enabling faster
type-checking and program execution.

\end{abstract}


\section{Dependency Analysis}

Consider this judgment from a type system that has been augmented with
\emph{dependency analysis}.

\[
    \ottmv{x} \! :^{  { \color{black}{L} }  }\! \ottkw{Int}  ,   \ottmv{y} \! :^{  { \color{black}{H} }  }\! \ottkw{Bool}   ,   \ottmv{z} \! :^{  { \color{black}{M} }  }\! \ottkw{Bool}    \vdash \,  \ottkw{if} \, \ottmv{z} \, \ottkw{then} \, \ottmv{x} \, \ottkw{else} \, \ottsym{3}  \, :^{  { \color{black}{M} }  } \,  \ottkw{Int} 
\]

In this judgment, $ { \color{black}{L} } $, $ { \color{black}{M} } $ and $ { \color{black}{H} } $ stand for low, medium and 
high security levels respectively. The computed value of the expression is meant to be a medium-security result. The inputs, $\ottmv{x}$, $\ottmv{y}$ and $\ottmv{z}$ have 
been marked with their respective security levels. This expression type-checks
because it is permissible for medium-security results to \emph{depend} on both
low and medium-security inputs. Note that the high-security boolean variable $\ottmv{y}$ is
not used in the expression. However, if we replace $\ottmv{z}$ with $\ottmv{y}$ in the
conditional, then the type checker would reject that expression. Even though the
high-security input would not be returned directly, the medium-security result would still  
depend on it.

Dependency analysis, as we see above, is an \emph{expressive} addition to
programming languages. Such analyses allow languages to protect
sensitive information~\citep{smith-volpano,slam}, support run-time code generation~\citep{binding-time}, slice programs while preserving behavior~\citep{slicing}, etc. Several existing dependency analyses were unified by \citet{dcc} in their Dependency Core Calculus (DCC). This calculus has served as a foundation
for static analysis of dependencies in programming languages.

What makes DCC powerful is the parameterization of the type system
by a \emph{generic} lattice of dependency levels. Dependency
analysis, in essence, is about ensuring secure information flow---that information never flows from more secure to less secure levels. \citet{denning1} showed that a lattice model, where increasing order corresponds to higher security, can be used to enforce secure flow of information. DCC integrates this lattice model with the computational $\lambda$-calculus \citep{moggi} by grading the monad operator of the latter with elements of the former. This integration enables DCC to analyze dependencies in its type system.

However, even though many typed languages have included dependency analysis in
some form, this feature has seen relatively little attention in the context of
\emph{dependently-typed} languages. This is unfortunate because, as we show in
this paper, dependency analysis can provide an elegant foundation for
compile-time and run-time irrelevance, two important concerns in the design of dependently-typed
languages. Compile-time irrelevance identifies sub-expressions that are not needed for type checking while run-time irrelevance identifies sub-expressions that do not affect the result of evaluation. 
By ignoring or erasing such sub-expressions, compilers for dependently-typed languages increase the expressiveness of the type system,
improve on compilation time and produce more efficient executables.

Therefore, in this work, we augment a dependently-typed language with a
\emph{primitive} notion of dependency analysis and use it to track compile-time and
run-time irrelevance. We call this language DDC, for Dependent Dependency Calculus, in homage to
DCC. Although our dependency analyses are structured differently,
we show that DDC can faithfully embed the terminating fragment of DCC and
support its many well-known applications, in addition to our
novel application of tracking compile-time and run-time irrelevance.

More specifically, our work makes the following contributions:
\begin{itemize}
\item We design a language SDC, for Simple Dependency Calculus, that can analyze 
  dependencies in a simply-typed language. We show that SDC is no less expressive
  than the terminating fragment of DCC. The structure of dependency analysis in SDC 
  enables a relatively straightforward syntactic proof of non-interference. (Section~\ref{SDC})
    
\item We extend SDC to a dependent calculus, $\textsc{DDC}^{\top}$. Using this calculus, we analyze
  run-time irrelevance and show the analysis is correct using a non-interference theorem. 
  $\textsc{DDC}^{\top}$ contains SDC as a sub-language. As such, it can be used to track other forms of dependencies as
  well.   (Section~\ref{DDCT})
\item We generalize $\textsc{DDC}^{\top}$ to DDC. Using this calculus, we analyze both run-time and
  compile-time irrelevance and show that the analyses are correct. To the best of our knowledge, DDC 
  is the only system that can distinguish run-time and compile-time irrelevance as separate modalities, necessary for the proper treatment  
  of projection from irrelevant $\Sigma$-types. (Section~\ref{sec:compile-time-irrelevance})
\item We have used the Coq proof assistant to mechanically verify the most
  important and delicate part of our designs, the non-interference and type
  soundness theorems for DDC. \ifanonymous Our proof scripts are available to
  reviewers as anonymized supplementary material and we plan to make this
  development publicly available. \else This mechanization is available online\footnote{
  \url{https://github.com/sweirich/graded-haskell}} and as a self-contained
  artifact~\citep{esop22:artifact}. \fi
\end{itemize}

\section{Irrelevance and Dependent Types}

\label{sec:irrelevance}

\textit{Run-time irrelevance} (sometimes called \emph{erasure}) and
\textit{compile-time irrelevance} are two forms of \textit{dependency} analyses that
arise in dependent type theories. Tracking these dependencies helps compilers produce faster
executables and makes type checking more flexible. \citep{pfenning:2001,miquel,barras:icc-star,mishra,abel,McBride:2016,atkey,Nuyts18,matus,Moon:2021}.






\subsection{Run-time irrelevance}

Parts of a program that are not required during run time are said to be
run-time irrelevant. Our goal is to identify such parts. Let's consider some examples. We shall mark variables and arguments
with $\top$ if they can be erased prior to execution and leave them unmarked if
they should be preserved. 

For example, the polymorphic identity function can be marked as:
\begin{lstlisting}
   id : Π x:$\!^{ { \color{black}{\top} } }$Type. x -> x 
   id = $\lambda\!^{ { \color{black}{\top} } }$x. λy. y
\end{lstlisting}
The first parameter, $\ottmv{x}$, of the identity function is only needed during
type checking; it can be erased before execution. The second
parameter, $\ottmv{y}$, though, is required during runtime.
When we apply this function to arguments, as in \lstinline{(id Bool$\!^{ { \color{black}{\top} } }$ True)}, we
can erase the first argument, \lstinline{Bool}, but the second one,
\lstinline{True}, must be retained.

Indexed data structures provide another example of run-time irrelevance. 

Consider the \cd{Vec} datatype for length-indexed
vectors, as it might look in a core
language inspired by GHC~\citep{systemfc,weirich:systemd}. The \cd{Vec} datatype has two parameters, \cd{n} and \cd{a}, that also appear in the types of the data constructors \cd{Nil} and \cd{Cons}. 
These parameters are relevant to \cd{Vec}, but irrelevant to the data constructors. (In the types of the constructors, the equality constraints
\cd{(n $\ \sim$ Zero)} and \cd{(n $\ \sim$ Succ m)} force \cd{n} to be equal to the length
of the vector.)

\begin{lstlisting}
Vec  : Nat -> Type -> Type
Nil  : Π n:$\!^{ { \color{black}{\top} } }$Nat. Π a:$\!^{ { \color{black}{\top} } }$Type. (n $\sim$ Zero) => Vec n a
Cons : Π n:$\!^{ { \color{black}{\top} } }$Nat. Π a:$\!^{ { \color{black}{\top} } }$Type. Π m:$\!^{ { \color{black}{\top} } }$Nat. (n $\sim$ Succ m) => a -> Vec m a -> Vec n a 
\end{lstlisting}

Now consider a function \cd{vmap} that maps a given function over a given vector. The length of the vector and the type arguments are not necessary for running \cd{vmap}; they are all erasable. So we assign them $ { \color{black}{\top} } $.

\noindent
\begin{minipage}{\linewidth}
\begin{lstlisting}
vmap : Π n:$\!^{ { \color{black}{\top} } }$Nat.Π a b:$\!^{ { \color{black}{\top} } }$Type.  (a -> b) -> Vec n a -> Vec n b
vmap = λ$\!^{ { \color{black}{\top} } }$ n a b. λ f xs.
          case xs of 
            Nil -> Nil 
            Cons m$\!^{ { \color{black}{\top} } }$ x xs -> Cons m$\!^{ { \color{black}{\top} } }$ (f x) (vmap m$\!^{ { \color{black}{\top} } }$ a$\!^{ { \color{black}{\top} } }$ b$\!^{ { \color{black}{\top} } }$ f xs)
\end{lstlisting}
\end{minipage}

Note that the $ { \color{black}{\top} } $-marked variables \cd{m}, \cd{a} and \cd{b} appear in the definition of \cd{vmap}, but only in $ { \color{black}{\top} } $ contexts. By requiring that `unmarked' terms \textit{don't depend} on terms marked with $ { \color{black}{\top} } $, we can track run-time irrelevance and guarantee safe erasure. Observe that even though
these arguments are marked with $ { \color{black}{\top} } $ to describe their use in the \emph{definition} of \cd{vmap}, this marking does not reflect their usage in the \emph{type} of \cd{vmap}. In particular, we are free to use these variables with \cd{Vec} in a relevant manner.

\subsection{Compile-time Irrelevance} \label{comp-irr}

Some type constructors may have arguments which can be ignored during type checking.
Such arguments are said to be \emph{compile-time irrelevant}. For example, suppose we have a constant function that ignores its argument and returns a type.

\begin{lstlisting}
phantom : Nat$\!^{ { \color{black}{\top} } }$ -> Type
phantom = λ$\!^{ { \color{black}{\top} } }$ x.  Bool
\end{lstlisting}

To type check \cd{idp} below, we must show that
\cd{phantom 0} equals \cd{phantom 1}. Without compile-time irrelevance,
we need to $\beta$-reduce both sides to show that the input and output types are equal.

\begin{lstlisting}
idp : phantom 0$\!^{ { \color{black}{\top} } }$ -> phantom 1$\!^{ { \color{black}{\top} } }$
idp = λ x. x
\end{lstlisting}

However, in the presence of compile-time irrelevance, we can use the dependency
information contained in the type of a function to reason about it abstractly.
Because the function \cd{f} below ignores its argument, it is sound to equate the input and output types.

\begin{lstlisting}
ida : Π f $:\!^{ { \color{black}{\top} } }$(Nat$\!^{ { \color{black}{\top} } }$ -> Type).  f 0$\!^{ { \color{black}{\top} } }$ -> f 1$\!^{ { \color{black}{\top} } }$
ida = λ$\!^{ { \color{black}{\top} } }$ f. λ x. x
\end{lstlisting}

In the absence of compile-time irrelevance, we cannot type-check \cd{ida}. So compile-time irrelevance makes type checking more flexible.

Compile-time irrelevance can also make type checking faster when the types contain
expensive computation that can be safely ignored. For example, consider the
following program that type checks without compile-time irrelevance.
However, in that case, the type checker must show that \cd{fib 28} reduces to
\cd{317811}, where \cd{fib} represents the Fibonacci function.

\begin{lstlisting}
idn : Π f $:\!^{ { \color{black}{\top} } }$(Nat$\!^{ { \color{black}{\top} } }$ -> Type). f (fib 28)$\!^{ { \color{black}{\top} } }$ -> f 317811$\!^{ { \color{black}{\top} } }$
idn = λ$\!^{ { \color{black}{\top} } }$ f. λ x. x
\end{lstlisting}

So far, we have used two annotations on variables and terms: $ { \color{black}{\top} } $ for irrelevant ones
and `unmarked' for relevant ones. We used $ { \color{black}{\top} } $ to mark both arguments that can be erased at runtime
and arguments that can be safely ignored by the type checker. However, sometimes we need a finer distinction.

\subsection{Strong Irrelevant $\Sigma$-types} 

Consider the type \cd{$\Sigma$m:$\!^{ { \color{black}{\top} } }$Nat. Vec m a}, which contains pairs whose
first component is marked as irrelevant. This type might be useful, say, for the output
of a \cd{filter} function for vectors, where the length of the output vector
cannot be calculated statically. If we never need to use this length at
runtime, then it would be good to mark it with $ { \color{black}{\top} } $ so that
it need not be stored.\footnote{It is, however, safe for \cd{m} to be used in a relevant position in the
  body of the $\Sigma$-type even when it is marked with $ { \color{black}{\top} } $. This marking indicates
  how the first component of a pair having this type is used, not how the bound variable \cd{m}
  is used in the body of the type.}

However, marking \cd{m} with $ { \color{black}{\top} } $ means that the first component of the pair of this type must also be
\emph{compile-time} irrelevant. This results in a significant
limitation for strong $\Sigma$ types: we cannot project the second component from the pair.
Say we have \cd{ys: $\Sigma$m:$\!^{ { \color{black}{\top} } }$Nat. Vec m a}. The type of (\cd{$\pi_1$ ys}) is 
a \cd{Nat} that can only be used in irrelevant positions. However, note that the argument \cd{n} in \cd{Vec n a} must be compile-time relevant; otherwise the type checker would equate \cd{Vec 0 a} with \cd{Vec 1  a}, making the length index meaningless. The 
type of (\cd{$\pi_2$ ys}) would then be \cd{Vec ($\pi_1$ ys) a}, 
which is ill-formed because an irrelevant term ($\pi_1$ ys) appears in a relevant position. 

Therefore, we don't want to mark the first component of the output of \cd{filter} with $ { \color{black}{\top} } $.
However, if we leave it unmarked, we cannot erase it at runtime, something that we might want to.
A way out of this quandry comes by considering terms that are run-time irrelevant but not compile-time irrelevant. Such terms exist between completely irrelevant and completely relevant terms. They should not depend upon irrelevant terms and relevant terms should not depend upon them. We mark such terms with a new annotation, $ { \color{black}{C} } $, with the constraints that `unmarked' terms do not depend on $ { \color{black}{C} } $ and $ { \color{black}{C} } $ terms do not depend on $ { \color{black}{\top} } $ terms. The three annotations, then, correspond to the three levels of a lattice modelling secure information flow, with $ { \color{black}{\bot} }  <  { \color{black}{C} }  <  { \color{black}{\top} } $, using $ { \color{black}{\bot} } $ in lieu of `unmarked'. We call the lattice $\mathcal{L}_I$, for irrelevance lattice. 
Using this lattice, we can type check the following \cd{filter} function.

\noindent
\begin{minipage}{\linewidth}
\noindent
\begin{lstlisting}
filter : Πn:$\!^{ { \color{black}{\top} } }$Nat.Πa:$\!^{ { \color{black}{\top} } }$Type.(a -> Bool) -> Vec n a -> $\Sigma$m:$^{ { \color{black}{C} } }$Nat. Vec m a
filter = λ$\!^{ { \color{black}{\top} } }$ n a. λ f vec.
          case vec of  
            Nil -> (Zero$^{ { \color{black}{C} } }$, Nil)  
            Cons n1$\!^{ { \color{black}{\top} } }$ x xs 
              | f x  -> ((Succ ($\pi_1$ ys))$^{ { \color{black}{C} } }$, Cons ($\pi_1$ ys)$\!^{ { \color{black}{\top} } }$ x ($\pi_2$ ys))
                            where 
                              ys = filter n1$\!^{ { \color{black}{\top} } }$ a$\!^{ { \color{black}{\top} } }$ f xs
              | _    -> filter n1$\!^{ { \color{black}{\top} } }$ a$\!^{ { \color{black}{\top} } }$ f xs
\end{lstlisting}
\end{minipage}

\citet{eisenberg:existentials} observe that, in Haskell, it is important to use
projection functions to access the components of the pair that
results from the recursive call (as in \cd{$\pi_1$ ys} and \cd{$\pi_2$ ys}) to
ensure that \cd{filter} is not excessively strict.  If \cd{filter} instead used pattern matching to eliminate the pair returned by the recursive call, it would
have needed to filter the entire vector before returning the first successful
value. This \cd{filter} function demonstrates the practical utility of strong irrelevant $\Sigma$-types because it supports the same run-time behavior of the usual list \cd{filter} function but with a more richly-typed data structure.

\section{A Simple Dependency Analyzing Calculus} \label{SDC}


Our ultimate goal is a dependent dependency calculus. However, we first start with a simply-typed version
so that we can explain our approach to dependency analysis and non-interference in a simplified setting. 

We call the calculus of this section SDC, for Simple Dependency Calculus.  This calculus is parameterized by a lattice of \emph{labels} or \emph{grades}, which can also be thought of as security \emph{levels}.\footnote{We use the terms label, level and grade interchangeably.} 
An excerpt of this calculus appears in Figure~\ref{fig:min}; it 
is an extension of the simply-typed $\lambda$-calculus with a grade-indexed
modal type $ T^{ \ell }\;  \ottnt{A} $. The modal type $ T^{ \ell }\;  \ottnt{A} $ can be thought of as putting a security barrier of grade $\ell$ around the values of $\ottnt{A}$. The calculus itself is also \emph{graded}, which means that in a typing judgment, the
derived term and every variable in the context carries a label or grade.
(The specification of the full system, which includes unit, products and sums, appears in \auxref{app:sdc-rules}.) 

\begin{figure}
\centering
\begin{flushright}
\textit{(Grammar)}
\end{flushright}
\vspace{-3ex}
\[
\begin{array}{llcll}
\textit{labels} & \ell, k & ::= &  { \color{black}{\bot} }  \mid  { \color{black}{\top} }  \mid  k  \wedge  \ell  \mid  k  \vee  \ell  \mid \ldots \\ 
\textit{types} & \ottnt{A},\ottnt{B}   & ::=& \ottkw{Unit} \mid   \ottnt{A}  \to  \ottnt{B}   \mid  T^{ \ell }\;  \ottnt{A} \\
\textit{terms} & \ottnt{a}, \ottnt{b}   & ::=& \ottmv{x} \mid  \lambda \ottmv{x} \!:\! \ottnt{A} . \ottnt{a}  \mid  \ottnt{a}  \;  \ottnt{b}  & \mbox{\it variables and functions} \\ 
                              && \mid &  \eta^{ \ell }\;  \ottnt{a}  \mid  \ottkw{bind} ^{ \ell } \,  \ottmv{x}  =  \ottnt{a}  \,  \ottkw{in}  \,  \ottnt{b}  & \mbox{\it graded modality}
\\
\textit{contexts} & \Omega & ::= & \varnothing \mid  \Omega ,   \ottmv{x} \! :^{ \ell }\! \ottnt{A}  \\
\\
\end{array}
\]
\vspace{-6ex}
\drules[SDC]{$ \Omega  \vdash \,  \ottnt{a}  \, :^{ \ell } \,  \ottnt{A} $}{Typing rules}
{Var,Abs,App,Return,Bind}
\vspace{-6ex}
\drules[SDCStep]{$ \ottnt{a}  \leadsto  \ottnt{a'} $}{Small step}
{Beta,BindBeta}
\vspace{-3ex}
\caption{Simple Dependency Calculus (Excerpt)}
\label{fig:min}
\label{fig:typing}
\end{figure}

\subsection{Type System}

The typing judgment has the form $ \Omega  \vdash \,  \ottnt{a}  \, :^{ \ell } \,  \ottnt{A} $ which means that ``$\ell$ is allowed to
observe $\ottnt{a}$'' or that ``$\ottnt{a}$ is visible at $\ell$''. Selected typing rules for SDC appear in Figure \ref{fig:typing}. Most rules are straightforward and propagate the level of the sub-terms to the expression. 

The \rref{SDC-Var} requires that the grade of the variable in the
context must be less than or equal to the grade of the observer. In other
words, an observer at level $\ell$ is allowed to use a variable from level
$k$ if and only if $k  \leq  \ell$.  If the variable's level is too high, then
this rule does not apply, ensuring that information can always flow to more
secure levels but never to less secure ones.  Abstraction
\rref{SDC-Abs} uses the current level of the expression for the newly
introduced variable in the context. This makes sense because the argument to
the function is checked at the same level in \rref{SDC-App}.

The modal type, introduced and eliminated with \rref{SDC-Return} and
\rref{SDC-Bind} respectively, manipulates the levels. The former says that, if a
term is $( \ell  \vee  \ell_{{\mathrm{0}}} )$-secure, then we can put it in an
$\ell_{{\mathrm{0}}}$-secure box and release it at level $\ell$. An $\ell_{{\mathrm{0}}}$-secure
boxed term can be unboxed only by someone who has security clearance for 
$\ell_{{\mathrm{0}}}$, as we see in the latter rule. The join operation in \rref{SDC-Bind}
ensures that $\ottnt{b}$ can depend on $\ottnt{a}$ only if $\ottnt{b}$ itself is
$\ell_{{\mathrm{0}}}$-secure or $\ell_{{\mathrm{0}}}  \leq  \ell$.

\subsection{Meta-theoretic Properties} \label{sec:sdc-metatheory}


This type system satisfies the following properties related to levels.

First, we can always weaken our assumptions about the variables in the
context.  If a term is derivable with an assumption held at some grade, then
that term is also derivable with that assumption held at any lower
grade. Below, for any two contexts $\Omega_{{\mathrm{1}}} , \Omega_{{\mathrm{2}}}$, we say that
$\Omega_{{\mathrm{1}}}  \leq  \Omega_{{\mathrm{2}}}$ iff they are the same modulo the grades and further, for
any $\ottmv{x}$, if $ \ottmv{x} \! :^{ \ell_{{\mathrm{1}}} }\! \ottnt{A}  \in  \Omega_{{\mathrm{1}}} $ and $ \ottmv{x} \! :^{ \ell_{{\mathrm{2}}} }\! \ottnt{A}  \in  \Omega_{{\mathrm{2}}} $, then
$\ell_{{\mathrm{1}}}  \leq  \ell_{{\mathrm{2}}}$.

\begin{lemma}[Narrowing]
If $ \Omega'  \vdash \,  \ottnt{a}  \, :^{ \ell } \,  \ottnt{A} $ and $\Omega  \leq  \Omega'$, then $ \Omega  \vdash \,  \ottnt{a}  \, :^{ \ell } \,  \ottnt{A} $.
\end{lemma}

Narrowing says that we can always downgrade any variable in the
context. Conversely, we cannot upgrade context variables in general, but we
can upgrade them to the level of the judgment.
\begin{lemma}[Restricted Upgrading] \label{pumping}
If $   \Omega_{{\mathrm{1}}} ,   \ottmv{x} \! :^{ \ell_{{\mathrm{0}}} }\! \ottnt{A}   ,  \Omega_{{\mathrm{2}}}   \vdash\,  \ottnt{b}  \, :^{ \ell }  \,  \ottnt{B} $ and $\ell_{{\mathrm{1}}}  \leq  \ell$, then $   \Omega_{{\mathrm{1}}} ,   \ottmv{x} \! :^{   \ell_{{\mathrm{0}}}  \vee  \ell_{{\mathrm{1}}}   }\! \ottnt{A}   ,  \Omega_{{\mathrm{2}}}   \vdash\,  \ottnt{b}  \, :^{ \ell }  \,  \ottnt{B} $.
\end{lemma}


The restricted upgrading lemma is needed to show subsumption. Subsumption states that, if a term is visible at some grade, then it is also visible at all higher grades.

\begin{lemma}[Subsumption] \label{subusage}
If $ \Omega  \vdash \,  \ottnt{a}  \, :^{ \ell } \,  \ottnt{A} $ and $\ell  \leq  k$, then $ \Omega  \vdash \,  \ottnt{a}  \, :^{ k } \,  \ottnt{A} $.
\end{lemma}

Subsumption is necessary (along with a standard weakening lemma) to show that substitution holds for this language. For substitution, we need to ensure that the level of the variable matches up with that of the substituted expression.
\begin{lemma}[Substitution] \label{substitution}
If $    \Omega_{{\mathrm{1}}} ,   \ottmv{x} \! :^{ \ell_{{\mathrm{0}}} }\! \ottnt{A}    ,  \Omega_{{\mathrm{2}}}   \vdash \,  \ottnt{b}  \, :^{ \ell } \,  \ottnt{B} $ and $ \Omega_{{\mathrm{1}}}  \vdash \,  \ottnt{a}  \, :^{ \ell_{{\mathrm{0}}} } \,  \ottnt{A} $, then $  \Omega_{{\mathrm{1}}} ,  \Omega_{{\mathrm{2}}}   \vdash \,  \ottnt{b}  \ottsym{\{}  \ottnt{a}  \ottsym{/}  \ottmv{x}  \ottsym{\}}  \, :^{ \ell } \,  \ottnt{B} $.
\end{lemma}

SDC terms are reduced using a call-by-name strategy. An excerpt of the small-step semantics appears in Figure \ref{fig:typing}. Note how the labels on the introduction form and the corresponding elimination form match up in \rref{SDCStep-Beta,SDCStep-BindBeta}. Further, note that we could have also used a call-by-value strategy to reduce SDC terms; we chose a call-by-name strategy because our development is motivated by potential applications in Haskell. 

For a call-by-name operational semantics, the above lemmas allow us to prove, a standard progress and preservation based type soundness result, which we omit here.  

Next, we show that our type system is secure by proving non-interference.

\subsection{A Syntactic Proof of Non-interference} 
\label{sec:geq}

When users with low-security clearance are oblivious to high-security data, we say that the system enjoys \emph{non-interference}. Non-interference results from level-specific views of the world. The values $ \eta^{  { \color{black}{H} }  }\;  \ottkw{True} $ and $ \eta^{  { \color{black}{H} }  }\;  \ottkw{False} $ appear the same to an $ { \color{black}{L} } $-user while an $ { \color{black}{H} } $-user can differentiate between them. To capture this notion of a level-specific view, we design an indexed equivalence relation on open terms, $\sim_{\ell}$, called \textit{indexed indistinguishability}, and shown in Figure \ref{fig:guarded-equality}. To define this relation, we need the labels of the variables in the context but not their types. So, we use grade-only contexts $\Phi$, defined as $\Phi ::= \varnothing \mid  \Phi  ,   \ottmv{x}  \! :  \ell  $. These contexts are like graded contexts $\Omega$ without the type information on variables, also denoted by $ |  \Omega  | $.  

\begin{figure}
\centering
\drules[SGEq]{$ \Phi  \vdash  \ottnt{a}  \sim_{ \ell }  \ottnt{b} $}{Indexed Indistinguishability}
{Var,Abs,App,Return,Bind}
\boxed{ \Phi  \vdash^{ \ell_{{\mathrm{0}}} }_{ \ell }  \ottnt{a_{{\mathrm{1}}}}  \sim  \ottnt{a_{{\mathrm{2}}}} }
\[ \drule{SEq-Leq} \qquad \drule{SEq-Nleq} \]
\caption{Indexed indistiguishability for SDC (Excerpt)}
\label{fig:guarded-equality}
\end{figure} 

Informally, $\Phi \vdash \ottnt{a} \sim_{\ell} \ottnt{b}$ means that $\ottnt{a}$ and $\ottnt{b}$ appear the same to an $\ell$-user. For example, $  \eta^{  { \color{black}{H} }  }\;  \ottkw{True}   \sim_{  { \color{black}{L} }  }   \eta^{  { \color{black}{H} }  }\;  \ottkw{False}  $ but $\neg(  \eta^{  { \color{black}{H} }  }\;  \ottkw{True}   \sim_{  { \color{black}{H} }  }   \eta^{  { \color{black}{H} }  }\;  \ottkw{False}  )$. We define this relation $\sim_{\ell}$ by structural induction on terms. We think of terms as ASTs annotated at various nodes with labels, say $\ell_{{\mathrm{0}}}$, that determine whether an observer $\ell$ is allowed to look at the corresponding sub-tree. If $\ell_{{\mathrm{0}}}  \leq  \ell$, then observer $\ell$ can start exploring the sub-tree; otherwise the entire sub-tree appears as a blob. So we can also read $\Phi \vdash \ottnt{a} \sim_{\ell} \ottnt{b}$ as: ``$\ottnt{a}$ is syntactically equal to $\ottnt{b}$ at all parts of the terms marked with any label $\ell_{{\mathrm{0}}}$, where $\ell_{{\mathrm{0}}}  \leq  \ell$, but may be 
arbitrarily different elsewhere.''


Note the \rref{SGEq-Return} in Figure \ref{fig:guarded-equality}. It uses an auxiliary relation, $ \Phi  \vdash^{ \ell_{{\mathrm{0}}} }_{ \ell }  \ottnt{a_{{\mathrm{1}}}}  \sim  \ottnt{a_{{\mathrm{2}}}} $. This auxiliary \textit{extended equivalence} relation $ \Phi  \vdash^{ \ell_{{\mathrm{0}}} }_{ \ell }  \ottnt{a_{{\mathrm{1}}}}  \sim  \ottnt{a_{{\mathrm{2}}}} $ formalizes the idea discussed above: if $\ell_{{\mathrm{0}}}  \leq  \ell$, then $\ottnt{a_{{\mathrm{1}}}}$ and $\ottnt{a_{{\mathrm{2}}}}$ must be indistinguishable at $\ell$; otherwise, they may be arbitrary terms. 

Now, we explore some properties of the indistinguishability relation.\footnote{Because this relation is untyped, its analogue for DDC is similar. For each lemma below, we include a reference to the location in the Coq development where it may be found for the dependent system.} \\ If we remove the second component from an indistinguishability relation, $ \Phi  \vdash  \ottnt{a}  \sim_{ \ell }  \ottnt{b} $, we get a new judgment, $ \Phi  \vdash  \ottnt{a}  :  \ell $, called grading judgment. Now, corresponding to every indistinguishability rule, we define a grading rule where the indistinguishability judgments have been replaced with their grading counterparts. Terms derived using these grading rules are called well-graded. We can show that well-typed terms are well-graded. 

\begin{lemma}[Typing implies grading] \label{typing_grading}
If $ \Omega  \vdash\,  \ottnt{a}  \, :^{ \ell }  \,  \ottnt{A} $ then $  |  \Omega  |   \vdash  \ottnt{a}  :  \ell $.
\end{lemma}

\begin{lemma}[Equivalence] \label{ind_equivalence}
Indexed indistinguishability at $\ell$ is an equivalence relation on well-graded terms at $\ell$.
\end{lemma}

The above lemma shows that indistinguishability is an equivalence relation. Observe that at the highest element of the lattice, $ { \color{black}{\top} } $, this equivalence degenerates to the identity relation.\\
Indistinguishability is closed under extended equivalence. The following is like a substitution lemma for the relation.

\begin{lemma}[Indistinguishability under substitution] 
  If $  \Phi  ,   \ottmv{x}  \! :  \ell    \vdash  \ottnt{b_{{\mathrm{1}}}}  \sim_{ k }  \ottnt{b_{{\mathrm{2}}}} $ and $ \Phi  \vdash^{ \ell }_{ k }  \ottnt{a_{{\mathrm{1}}}}  \sim  \ottnt{a_{{\mathrm{2}}}} $ then $ \Phi  \vdash  \ottnt{b_{{\mathrm{1}}}}  \ottsym{\{}  \ottnt{a_{{\mathrm{1}}}}  \ottsym{/}  \ottmv{x}  \ottsym{\}}  \sim_{ k }  \ottnt{b_{{\mathrm{2}}}}  \ottsym{\{}  \ottnt{a_{{\mathrm{2}}}}  \ottsym{/}  \ottmv{x}  \ottsym{\}} $.
\end{lemma}

With regard to the above lemma, consider the situation when $ \neg  \ottsym{(}  \ell  \leq  k  \ottsym{)} $, for example, when $\ell =  { \color{black}{H} } $ and $k =  { \color{black}{L} } $. In such a situation, for any two terms $\ottnt{a_{{\mathrm{1}}}}$ and $\ottnt{a_{{\mathrm{2}}}}$, if $  \Phi  ,   \ottmv{x}  \! :  \ell    \vdash  \ottnt{b_{{\mathrm{1}}}}  \sim_{ k }  \ottnt{b_{{\mathrm{2}}}} $, then $ \Phi  \vdash  \ottnt{b_{{\mathrm{1}}}}  \ottsym{\{}  \ottnt{a_{{\mathrm{1}}}}  \ottsym{/}  \ottmv{x}  \ottsym{\}}  \sim_{ k }  \ottnt{b_{{\mathrm{2}}}}  \ottsym{\{}  \ottnt{a_{{\mathrm{2}}}}  \ottsym{/}  \ottmv{x}  \ottsym{\}} $. Let us work out a concrete example. For a typing derivation $  \ottmv{x} \! :^{  { \color{black}{H} }  }\! \ottnt{A}   \vdash\,  \ottnt{b}  \, :^{  { \color{black}{L} }  }  \,  \ottkw{Bool} $, we have, by lemmas \ref{typing_grading} and \ref{ind_equivalence}, $  \ottmv{x}  \! :   { \color{black}{H} }    \vdash  \ottnt{b}  \sim_{  { \color{black}{L} }  }  \ottnt{b} $. Then, $ \varnothing  \vdash  \ottnt{b}  \ottsym{\{}  \ottnt{a_{{\mathrm{1}}}}  \ottsym{/}  \ottmv{x}  \ottsym{\}}  \sim_{  { \color{black}{L} }  }  \ottnt{b}  \ottsym{\{}  \ottnt{a_{{\mathrm{2}}}}  \ottsym{/}  \ottmv{x}  \ottsym{\}} $. This is almost non-interference in action. What's left to show is that the indistinguishability relation respects the small step semantics, written $ \ottnt{a_{{\mathrm{1}}}}  \leadsto  \ottnt{a_{{\mathrm{2}}}} $. The small-step relation is standard call-by-name reduction. 

\begin{theorem}[Non-interference]
\label{lemma:GEq_respects_Step}
  If $ \Phi  \vdash  \ottnt{a_{{\mathrm{1}}}}  \sim_{ k }  \ottnt{a'_{{\mathrm{1}}}} $ and $ \ottnt{a_{{\mathrm{1}}}}  \leadsto  \ottnt{a_{{\mathrm{2}}}} $ then there exists some $\ottnt{a'_{{\mathrm{2}}}}$ such that $ \ottnt{a'_{{\mathrm{1}}}}  \leadsto  \ottnt{a'_{{\mathrm{2}}}} $ and $ \Phi  \vdash  \ottnt{a_{{\mathrm{2}}}}  \sim_{ k }  \ottnt{a'_{{\mathrm{2}}}} $.
\end{theorem}

Since the step relation is deterministic, in the above lemma, there is exactly one such $\ottnt{a'_{{\mathrm{2}}}}$ that $\ottnt{a'_{{\mathrm{1}}}}$ steps to. Now, going back to our last example, we see that $\ottnt{b}  \ottsym{\{}  \ottnt{a_{{\mathrm{1}}}}  \ottsym{/}  \ottmv{x}  \ottsym{\}}$ and $\ottnt{b}  \ottsym{\{}  \ottnt{a_{{\mathrm{2}}}}  \ottsym{/}  \ottmv{x}  \ottsym{\}}$ take steps in tandem and they are $ { \color{black}{L} } $-indistinguishable after each and every step. Since the language itself is terminating, both the terms reduce to boolean values, values that are themselves $ { \color{black}{L} } $-indistinguishable as well. But the indistinguishability for boolean values is just the identity relation. This means that $\ottnt{b}  \ottsym{\{}  \ottnt{a_{{\mathrm{1}}}}  \ottsym{/}  \ottmv{x}  \ottsym{\}}$ and $\ottnt{b}  \ottsym{\{}  \ottnt{a_{{\mathrm{2}}}}  \ottsym{/}  \ottmv{x}  \ottsym{\}}$ reduce to the same value.

The indistinguishability relation gives us a syntactic method of proving non-interference for programs derived in SDC. Essentially, we show that a user with low-security clearance cannot distinguish between high security values just by observing program behavior.\\ Next, we show that SDC is no less expressive than the terminating fragment of DCC.

\subsection{Relation with Sealing Calculus and Dependency Core Calculus}

SDC is extremely similar to the sealing calculus $\lambda^{[]}$ of
\citet{igarashi}. Like SDC, $\lambda^{[]}$ has a label on the typing
judgment.\footnote{Note that our labels correspond to observer levels of
\citet{igarashi}, which can be viewed as a lattice.} But unlike SDC, $\lambda^{[]}$ uses
standard ungraded typing contexts $\Gamma$. Both the calculi have the same types. As far as terms are concerned, there is only one difference. The sealing calculus has an $\ottkw{unseal}$ term whereas SDC uses
$\ottkw{bind}$. We present the rules for sealing and unsealing terms in $\lambda^{[]}$ below.\footnote{We take the liberty of making small cosmetic changes in the presentation.}

\[ \drule{Sealing-Seal}\drule{Sealing-Unseal} \]

\citet{igarashi} have shown that $\lambda^{[]}$ is equivalent to
$\text{DCC}_{\text{pc}}$, an extension of the terminating fragment of
DCC. Therefore, we compare SDC to DCC by simulating $\lambda^{[]}$ in SDC.
For this, we define a translation $\bar{\cdot}$, from $\lambda^{[]}$ to SDC. Most of the cases are handled inductively in a straightforward manner. For $\ottkw{unseal}$, we have, $ \overline{   \ottkw{unseal}^{ \ell }  \ottnt{a}   }  :=  \ottkw{bind} ^{ \ell } \,  \ottmv{x}  =    \overline{ \ottnt{a} }    \,  \ottkw{in}  \,  \ottmv{x} $.

With this translation, we can give a forward and a backward simulation
connecting the two languages. The reduction relation $\leadsto$ below is
full reduction for both the languages, the reduction strategy used by \citet{igarashi} for $\lambda^{[]}$. Full reduction is a non-deterministic reduction strategy whereby a $\beta$-redex in any sub-term may be reduced.

\begin{theorem}[Forward Simulation]
If $ \ottnt{a}  \leadsto  \ottnt{a'} $ in $\lambda^{[]}$, then $  \overline{ \ottnt{a} }   \leadsto   \overline{ \ottnt{a'} }  $ in SDC.
\end{theorem}

\begin{theorem}[Backward Simulation]
For any term $a$ in $\lambda^{[]}$, if $  \overline{ \ottnt{a} }   \leadsto  \ottnt{b} $ in SDC, then there exists $\ottnt{a'}$ in $\lambda^{[]}$ such that $\ottnt{b}  \ottsym{=}   \overline{ \ottnt{a'} } $ and $ \ottnt{a}  \leadsto  \ottnt{a'} $.
\end{theorem}

The translation also preserves typing. In fact, a source term and its target have the same type. Below, for an ordinary context $\Gamma$, the graded context $ \Gamma ^{ \ell } $ denotes $\Gamma$ with the labels
for all the variables set to $\ell$.

\begin{theorem}[Translation Preserves Typing]
If $ \Gamma   \vdash  \ottnt{a}  :^{ \ell }  \ottnt{A} $, then $  \Gamma ^{ \ell }   \vdash\,   \overline{ \ottnt{a} }   \, :^{ \ell }  \,  \ottnt{A} $.
\end{theorem}

The above translation shows that the terminating fragment of DCC can be embedded into SDC. Therefore SDC is at least as expressive as the terminating fragment of DCC. Further, SDC lends itself nicely to syntactic proof techniques for non-interference. This approach generalizes to more expressive systems, as we shall see in the next section, where we extend SDC to a general dependent dependency calculus.

\section{A Dependent Dependency Analyzing Calculus} \label{DDCT}

\begin{figure}
\centering
\[
\begin{array}{lcll}
\ottnt{a}, \ottnt{A}, \ottnt{b}, \ottnt{B}   
   & ::=  &  \ottnt{s}  \mid \ottkw{unit} \mid \ottkw{Unit} & \mbox{\it sorts and unit }\\
   & \mid &  \Pi  \ottmv{x} \!:^{ \ell }\! \ottnt{A} . \ottnt{B}  \mid \ottmv{x} \mid 
             \lambda  \ottmv{x} \!:^{ \ell }\! \ottnt{A}  .  \ottnt{a}  \mid  \ottnt{a}  \;  \ottnt{b} ^{ \ell }  & \mbox{\it dependent functions} \\
   & \mid &  \Sigma  \ottmv{x} \!\!:^{ \ell }\!\! \ottnt{A} . \ottnt{B}  \mid 
             ( \ottnt{a} ^{ \ell },  \ottnt{b} )  \mid  \ottkw{let} \; ( \ottmv{x} ^{ \ell } ,  \ottmv{y} )\ =\  \ottnt{a} \  \ottkw{in} \  \ottnt{b}  & \mbox{\it dependent pairs} \\
   & \mid &  \ottnt{A}  +  \ottnt{B}  \mid  \ottkw{inj}_1\,  \ottnt{a}  \mid  \ottkw{inj}_2\,  \ottnt{a}  \mid 
             \ottkw{case} \,  \ottnt{a} \, \ottkw{of}\,  \ottnt{b_{{\mathrm{1}}}}  ;  \ottnt{b_{{\mathrm{2}}}}  & \mbox{\it disjoint unions} \\
\\
\end{array}
\]
\caption{Dependent Dependency Calculus Grammar (Types and Terms)}
\label{fig:ddc-grammar}
\end{figure}

Here and in the next section, we present dependently-typed languages, with dependency analysis in the style of SDC.
The first extension, called $\textsc{DDC}^{\top}$ is a straightforward integration of labels and dependent types. This 
system subsumes SDC, and so can be used for the same purposes. Here, we show how it can be used to analyze \emph{run-time irrelevance}. 
Then, in Section \ref{sec:compile-time-irrelevance}, we generalize this system to \DDC{}, which allows definitional equality to ignore unnecessary sub-terms, thus also enabling \emph{compile-time irrelevance}. We present the system in this way both to simplify the
presentation and to show that $\textsc{DDC}^{\top}$ is an intermediate point in the design space.

Both $\textsc{DDC}^{\top}$ and \DDC{} are pure type systems~\citep{pts}.
They share the same syntax, shown in Figure~\ref{fig:ddc-grammar}, combining 
terms and types into the same grammar.
They are parameterized by a set of sorts $\ottnt{s}$, a set of axioms $ \mathcal{A}( \ottnt{s_{{\mathrm{1}}}} , \ottnt{s_{{\mathrm{2}}}} ) $ which is a binary relation on sorts, and a set of rules $ \mathcal{R}( \ottnt{s_{{\mathrm{1}}}} , \ottnt{s_{{\mathrm{2}}}} , \ottnt{s_{{\mathrm{3}}}} ) $ which is a ternary relation on sorts. For simplicity, we assume, without loss of generality, that for every sort $\ottnt{s_{{\mathrm{1}}}}$, there is some sort $\ottnt{s_{{\mathrm{2}}}}$, such that $ \mathcal{A}( \ottnt{s_{{\mathrm{1}}}} , \ottnt{s_{{\mathrm{2}}}} ) $.\footnote{This assumption does not lead to any loss in generality because given a pure type system $(\mathit{S'},\mathit{A'},\mathit{R'})$ that does not meet the above condition, we can provide another pure type system $(\mathit{S''}, \mathit{A''}, \mathit{R''})$, where $\mathit{S''} = S' \cup \{ \pentagon \}$ (given $\pentagon \notin S'$) and $\mathit{A''} = A' \cup \{ (s, \pentagon) | s \in S'' \}$ and $\mathit{R''} = \mathit{R'}$, such that there exists a straightforward bisimulation between the two systems.}

We annotate several syntactic forms with grades for dependency analysis.
The dependent function type, written $ \Pi  \ottmv{x} \!:^{ \ell }\! \ottnt{A} . \ottnt{B} $, includes the grade of the
argument to a function having this type. Similarly, the dependent pair type, written
$ \Sigma  \ottmv{x} \!\!:^{ \ell }\!\! \ottnt{A} . \ottnt{B} $, includes the grade of the first component of a pair having this type. \footnote{We use standard abbreviations when $\ottmv{x}$ is not free in $\ottnt{B}$: we write $ {}^{ \ell } \!  \ottnt{A}  \to  \ottnt{B} $
for $ \Pi  \ottmv{x} \!:^{ \ell }\! \ottnt{A} . \ottnt{B} $ and $ {}^{ \ell } \!  \ottnt{A}  \times  \ottnt{B} $ for $ \Sigma  \ottmv{x} \!\!:^{ \ell }\!\! \ottnt{A} . \ottnt{B} $.} We can interpret these types as a fusion of the usual, ungraded dependent types and the graded modality $ T^{ \ell }\;  \ottnt{A} $ we saw earlier. In other words, $ \Pi  \ottmv{x} \!:^{ \ell }\! \ottnt{A} . \ottnt{B} $ acts like the type
$ \Pi \ottmv{y} \!:\! \ottsym{(}   T^{ \ell }\;  \ottnt{A}   \ottsym{)} .  \ottkw{bind}  \,  \ottmv{x}  =  \ottmv{y}  \,  \ottkw{in}  \,  \ottnt{B}  $ and $ \Sigma  \ottmv{x} \!\!:^{ \ell }\!\! \ottnt{A} . \ottnt{B} $ acts like the type
$ \Sigma \ottmv{y} \!:\! \ottsym{(}   T^{ \ell }\;  \ottnt{A}   \ottsym{)} .  \ottkw{bind}  \,  \ottmv{x}  =  \ottmv{y}  \,  \ottkw{in}  \,  \ottnt{B}  $. Because of this fusion, we do not need
to add the graded modality type as a separate form---we can define $ T^{ \ell }\;  \ottnt{A} $ as $ \Sigma  \ottmv{x} \!\!:^{ \ell }\!\! \ottnt{A} . \ottkw{Unit} $.  Using $ \Pi  \ottmv{x} \!:^{ \ell }\! \ottnt{A} . \ottnt{B} $ instead of $ \Pi \ottmv{y} \!:\! \ottsym{(}   T^{ \ell }\;  \ottnt{A}   \ottsym{)} .  \ottkw{bind}  \,  \ottmv{x}  =  \ottmv{y}  \,  \ottkw{in}  \,  \ottnt{B}  $ has an advantage:
the former allows $\ottmv{x}$ to be held at differing grades while type checking $\ottnt{B}$ and the body of a function having this $\Pi$-type while the latter requires $\ottmv{x}$ to be held at the same grade in both the cases. We utilize this flexibility in Section \ref{sec:compile-time-irrelevance}. 


\subsection{$\textsc{DDC}^{\top}$ : $\Pi$-types}
\label{sec:erasure}

\begin{figure}
\begin{drulepar*}[]{$ \Omega  \vdash  \ottnt{a} :^{ \ell }  \ottnt{A} $}{Typing}
  \drule{DCT-Var}
  \drule{DCT-Type}
  \drule[width=2.5in]{DCT-Pi}
  \drule{DCT-Abs}
  \drule{DCT-App}
  \drule{DCT-Conv}
\end{drulepar*}
\caption{$\textsc{DDC}^{\top}$ type system (core rules)}
\label{fig:sdc-typing}
\end{figure}

The core typing rules for $\textsc{DDC}^{\top}$ appear in
Figure~\ref{fig:sdc-typing}. As in the simple type system, the variables in
the context are labelled and the judgement itself includes a label $\ell$.
\Rref{DCT-Var} is similar to its counterpart in the simply-typed language: the variable being observed must be graded less than or equal to the level of the observer. 
\Rref{DCT-Pi} propagates the level of the expression to the sub-terms of the
$\Pi$-type.  Note that this type is annotated with an arbitrary label
$\ell_{{\mathrm{0}}}$: the purpose of this label $\ell_{{\mathrm{0}}}$ is to denote the level at
which the argument to a function having this type may be used.

In \rref{DCT-Abs}, the parameter of the function is introduced into the context
at level $ \ell_{{\mathrm{0}}}  \vee  \ell $ (akin to \rref{SDC-Bind}). In \rref{DCT-App}, the
argument to the function is checked at level $ \ell_{{\mathrm{0}}}  \vee  \ell $ (akin to \rref{SDC-Return}).
Note that the $\Pi$-type is checked at $ { \color{black}{\top} } $ in \rref{DCT-Abs}. In
$\textsc{DDC}^{\top}$, level $ { \color{black}{\top} } $ corresponds to `compile time' observers
and motivates the superscript $ { \color{black}{\top} } $ in the language name. 

\Rref{DCT-Conv} converts the type of an expression to an equivalent type.  The
judgment $  |  \Omega  |   \vdash  \ottnt{A}  \equiv_{  { \color{black}{\top} }  }  \ottnt{B} $ is a label-indexed definitional
equality relation instantiated to $ { \color{black}{\top} } $. This relation is the closure of the indexed indistinguishability relation (Section \ref{sec:geq}) under small-step call-by-name evaluation. When instantiated to $ { \color{black}{\top} } $, the relation degenerates to $\beta$-equivalence. So the \rref{DCT-Conv} is essentially casting a term to a $\beta$-equivalent type; however, in the next section, we utilize the flexibility of label-indexing to cast a term to a type that may not be $\beta$-equivalent. Also, note that the equality relation itself is untyped. As such, we need the third premise to guarantee that the new type is well-formed.

\subsection{$\textsc{DDC}^{\top}$ : $\Sigma$-types}

The language $\textsc{DDC}^{\top}$ includes $\Sigma$ types, as specified by the rules below.

\[ \drule[width=3in]{DCT-WSigma} \ \drule[width=3in]{DCT-WPair} \]


Like $\Pi$-types, $\Sigma$-types include a grade that is not related to how the bound variable is used in the body of the type. The grade indicates the level at which the first component of a pair having the $\Sigma$-type may be used.
In \rref{DCT-WPair}, we check the first component $\ottnt{a}$ of the
pair at a level raised by $\ell_{{\mathrm{0}}}$, the level annotating the
type, akin to \rref{SDC-Return}. The second component $\ottnt{b}$ is checked at the current level.

\[ \drule[width=5in]{DCT-LetPair} \]

The \rref{DCT-LetPair} eliminates pairs using dependently-typed
pattern matching. The pattern variables $\ottmv{x}$ and $\ottmv{y}$
are introduced into the context while checking the body $\ottnt{c}$. Akin to \rref{SDC-Bind},
the level of the first pattern variable, $\ottmv{x}$, is raised by $\ell_{{\mathrm{0}}}$. 
The result type $\ottnt{C}$ is refined by the
pattern match, informing the type system that the pattern $ ( \ottmv{x} ^{ \ell_{{\mathrm{0}}} },  \ottmv{y} ) $ is equal
to the scrutinee $\ottnt{a}$.

Because of this refinement in the result type, we can define
the projection operations through pattern matching. In particular, the first projection, $ \pi_1^{ \ell_{{\mathrm{0}}} }\; \ottnt{a}  :=  \ottkw{let} \; ( \ottmv{x} ^{ \ell_{{\mathrm{0}}} } ,  \ottmv{y} )\ =\  \ottnt{a} \  \ottkw{in} \  \ottmv{x} $ while the second projection, $ \pi_2^{ \ell_{{\mathrm{0}}} }  \ottnt{a}  :=  \ottkw{let} \; ( \ottmv{x} ^{ \ell_{{\mathrm{0}}} } ,  \ottmv{y} )\ =\  \ottnt{a} \  \ottkw{in} \  \ottmv{y} $. These projections can be type checked
according to the following derived rules:

\[ \drule[width=3in]{DCT-ProjOne} \qquad \drule{DCT-ProjTwo} \]

Note that the derived \rref{DCT-Proj1} limits access to the first
component through the premise $\ell_{{\mathrm{0}}}  \leq  \ell$, akin to \rref{Sealing-Unseal}.  This condition makes sense
because it aligns the observability of the first
component of the pair with the label on the $\Sigma$-type.


\subsection{Embedding SDC into $\textsc{DDC}^{\top}$}

Here, we show how to embed SDC into $\textsc{DDC}^{\top}$.

We define a translation function, $\overline{\overline{\cdot}}$, that takes the types and terms in SDC to terms in $\textsc{DDC}^{\top}$. For types, the translation is defined as: $ \overline{\overline{   \ottnt{A}  \to  \ottnt{B}   } }  :=  {}^{  { \color{black}{\bot} }  } \!    \overline{\overline{ \ottnt{A} } }    \to    \overline{\overline{ \ottnt{B} } }   $, $ \overline{\overline{   \ottnt{A}  \times  \ottnt{B}   } }  :=  {}^{  { \color{black}{\bot} }  } \!    \overline{\overline{ \ottnt{A} } }    \times    \overline{\overline{ \ottnt{B} } }   $ and $ \overline{\overline{   T^{ \ell }\;  \ottnt{A}   } }  :=  \Sigma  \ottmv{x} \!\!:^{ \ell }\!\!  \overline{\overline{ \ottnt{A} } }  . \ottkw{Unit} $. For terms, the translation is straightforward except for the following cases: $ \overline{\overline{   \eta^{ \ell }\;  \ottnt{a}   } }  :=  (  \overline{\overline{ \ottnt{a} } }  ^{ \ell },  \ottkw{unit} ) $ and $ \overline{\overline{   \ottkw{bind} ^{ \ell } \,  \ottmv{x}  =  \ottnt{a}  \,  \ottkw{in}  \,  \ottnt{b}   } }  :=  \ottkw{let} \; ( \ottmv{x} ^{ \ell } ,  \ottmv{y} )\ =\   \overline{\overline{ \ottnt{a} } }  \  \ottkw{in} \   \overline{\overline{ \ottnt{b} } }  $, where $\ottmv{y}$ is a fresh variable. 
By lifting the translation to contexts, we show that translation preserves typing.
\begin{theorem}[Trans. Preserves Typing]
If $ \Omega  \vdash \,  \ottnt{a}  \, :^{ \ell } \,  \ottnt{A} $, then $  \overline{\overline{ \Omega } }   \vdash   \overline{\overline{ \ottnt{a} } }  :^{ \ell }   \overline{\overline{ \ottnt{A} } }  $.
\end{theorem} 

Next, assuming a standard call-by-name small-step semantics for both the languages, we can provide a bisimulation.
\begin{theorem}[Forward Simulation]
If $ \ottnt{a}  \leadsto  \ottnt{a'} $ in SDC, then $  \overline{\overline{ \ottnt{a} } }   \leadsto   \overline{\overline{ \ottnt{a'} } }  $ in $\textsc{DDC}^{\top}$.
\end{theorem}
\begin{theorem}[Backward Simulation]
For any term $a$ in SDC, if $  \overline{\overline{ \ottnt{a} } }   \leadsto  \ottnt{b} $ in $\textsc{DDC}^{\top}$, then there exists $\ottnt{a'}$ in SDC such that $\ottnt{b}  \ottsym{=}   \overline{\overline{ \ottnt{a'} } } $ and $ \ottnt{a}  \leadsto  \ottnt{a'} $.
\end{theorem}

Hence, SDC can be embedded into $\textsc{DDC}^{\top}$, preserving meaning. As such, $\textsc{DDC}^{\top}$ can analyze dependencies in general.

\subsection{Run-time Irrelevance}

Next, we show how to track run-time irrelevance using $\textsc{DDC}^{\top}$. We use the two element lattice $\{  { \color{black}{\bot} }  ,  { \color{black}{\top} }  \}$ with $ { \color{black}{\bot} }  <  { \color{black}{\top} } $ such that $ { \color{black}{\bot} } $ and $ { \color{black}{\top} } $ correspond to run-time relevant and run-time irrelevant terms respectively. So, we need to erase terms marked with $ { \color{black}{\top} } $. However, we first define a general indexed erasure function, $\lfloor \cdot \rfloor_{\ell}$, on $\textsc{DDC}^{\top}$ terms, that erases everything an $\ell$-user should not be able to see. The function is defined by straightforward recursion in most cases. For example, \\ $ \lfloor  \ottmv{x}  \rfloor_ \ell  := \ottmv{x}$ and $ \lfloor   \Pi  \ottmv{x} \!:^{ \ell_{{\mathrm{0}}} }\! \ottnt{A} . \ottnt{B}   \rfloor_ \ell  :=  \Pi  \ottmv{x} \!:^{ \ell_{{\mathrm{0}}} }\!  \lfloor  \ottnt{A}  \rfloor_ \ell  .  \lfloor  \ottnt{B}  \rfloor_ \ell  $ and $ \lfloor   \lambda^{ \ell_{{\mathrm{0}}} }  \ottmv{x} . \ottnt{b}   \rfloor_ \ell  :=  \lambda^{ \ell_{{\mathrm{0}}} }  \ottmv{x} .  \lfloor  \ottnt{b}  \rfloor_ \ell  $. \\ The interesting cases are: \\  $\qquad  \lfloor    \ottnt{b}  \;  \ottnt{a} ^{ \ell_{{\mathrm{0}}} }    \rfloor_ \ell  := \ottsym{(}     \lfloor  \ottnt{b}  \rfloor_ \ell    \;   \lfloor  \ottnt{a}  \rfloor_ \ell  ^{ \ell_{{\mathrm{0}}} }   \ottsym{)}$ if $\ell_{{\mathrm{0}}}  \leq  \ell$ and $\ottsym{(}     \lfloor  \ottnt{b}  \rfloor_ \ell    \;  \ottkw{unit} ^{ \ell_{{\mathrm{0}}} }   \ottsym{)}$ otherwise, \\  $\qquad  \lfloor   ( \ottnt{a} ^{ \ell_{{\mathrm{0}}} },  \ottnt{b} )   \rfloor_ \ell  :=  (   \lfloor  \ottnt{a}  \rfloor_ \ell   ^{ \ell_{{\mathrm{0}}} },   \lfloor  \ottnt{b}  \rfloor_ \ell  ) $ if $\ell_{{\mathrm{0}}}  \leq  \ell$ and $ ( \ottkw{unit} ^{ \ell_{{\mathrm{0}}} },   \lfloor  \ottnt{b}  \rfloor_ \ell  ) $ otherwise. \\ They are so defined because if $ \neg  \ottsym{(}  \ell_{{\mathrm{0}}}  \leq  \ell  \ottsym{)} $, an $\ell$-user should not be able to see $\ottnt{a}$, so we replace it with $\ottkw{unit}$.

This erasure function is closely related to the indistinguishability relation, we saw in Section \ref{sec:geq}, extended to a dependent setting. (This definition appears in \auxref{dep-indist}.) The erasure function maps the equivalence classes formed by the indistinguishability relation to their respective canonical elements. We have verified the following lemmas using the Coq proof assistant. Footnotes mark the file and lemma name of the corresponding mechanized results.
\begin{lemma}[Canonical Element\footnote{\texttt{erasure.v:Canonical\_element}.}]
If $ \Phi  \vdash  \ottnt{a_{{\mathrm{1}}}}  \sim_{ \ell }  \ottnt{a_{{\mathrm{2}}}} $, then $ \lfloor  \ottnt{a_{{\mathrm{1}}}}  \rfloor_ \ell  =  \lfloor  \ottnt{a_{{\mathrm{2}}}}  \rfloor_ \ell $.
\end{lemma}

Further, a well-graded term and its erasure are indistinguishable.
\begin{lemma}[Erasure Indistinguishability\footnote{\texttt{erasure.v:Erasure\_Indistinguishability}}] \label{erasure_ind}
If $ \Phi  \vdash  \ottnt{a}  :  \ell $, then $ \Phi  \vdash  \ottnt{a}  \sim_{ \ell }   \lfloor  \ottnt{a}  \rfloor_ \ell  $.
\end{lemma}

Next, we can show that erased terms simulate the reduction behavior of their unerased counterparts.

\begin{lemma}[Erasure Simulation\footnote{\texttt{erasure.v:Step\_erasure,Value\_erasure}}]
If $ \Phi  \vdash  \ottnt{a}  :  \ell $ and $ \ottnt{a}  \leadsto  \ottnt{b} $, then $  \lfloor  \ottnt{a}  \rfloor_ \ell   \leadsto   \lfloor  \ottnt{b}  \rfloor_ \ell  $. Otherwise, if $\ottnt{a}$ is a value, then so is $ \lfloor  \ottnt{a}  \rfloor_ \ell $.
\end{lemma}

This lemma follows from Lemma \ref{erasure_ind} and the non-interference theorem (Theorem \ref{lemma:GEq_respects_Step}). Therefore, it is safe to erase, before run time, all sub-terms marked with $ { \color{black}{\top} } $.

This shows that we can correctly analyze run-time irrelevance using $\textsc{DDC}^{\top}$. However, supporting compile-time irrelevance requires some changes to the system. We take them up in the next section.

\section{\DDC{}: Run-time and Compile-time Irrelevance}
\label{sec:compile-time-irrelevance}

\subsection{Towards Compile-time Irrelevance} 


Recall that terms which may be safely ignored while checking for type equality are said to be compile-time irrelevant. In $\textsc{DDC}^{\top}$, the conversion \rref{DCT-Conv} checks for type equality at $ { \color{black}{\top} } $. 
\[ \drule[width=3in]{DCT-Conv} \]
The equality judgment used in this rule $ \Phi  \vdash  \ottnt{a}  \equiv_{  { \color{black}{\top} }  }  \ottnt{b} $ is an instantiation of the general judgment $ \Phi  \vdash  \ottnt{a}  \equiv_{ \ell }  \ottnt{b} $, which is the closure of the indistinguishability relation at $\ell$ under $\beta$-equivalence. When $\ell$ is $ { \color{black}{\top} } $, indistinguishability is just identity. As such, the equality relation at $ { \color{black}{\top} } $ degenerates to standard $\beta$-equivalence. So, \rref{DCT-Conv} does not ignore any part of the terms when checking for type equality. 

To support compile-time irrelevance then, we need the conversion rule to use equality at some grade strictly less than $ { \color{black}{\top} } $ so that $ { \color{black}{\top} } $-marked terms may be ignored. For the irrelevance lattice $\mathcal{L}_I$, the level $ { \color{black}{C} } $ can be used for this purpose. For any other lattice $\mathcal{L}$, we can add two new elements, $ { \color{black}{C} } $ and $ { \color{black}{\top} } $, above every other existing element, such that $\mathcal{L} <  { \color{black}{C} }  <  { \color{black}{\top} } $, and thereafter use level $ { \color{black}{C} } $ for this purpose. So, for any lattice, we can support compile-time irrelevance by equating types at $ { \color{black}{C} } $.
 
Referring back to the examples in Section \ref{comp-irr}, note that for \cd{phantom : Nat$\!^{ { \color{black}{\top} } }$ -> Type}, we have \cd{phantom 0$^{ { \color{black}{\top} } }$ $\equiv_{ { \color{black}{C} } }$ phantom 1$^{ { \color{black}{\top} } }$}. With this equality, we can type-check \cd{idp : phantom 0$^{ { \color{black}{\top} } }$ -> phantom 1$^{ { \color{black}{\top} } }$ = λ x. x}, even without knowing the definition of \cd{phantom}.

Now, observe that in \rref{DCT-Conv}, the new type $\ottnt{B}$ is also checked at $ { \color{black}{\top} } $. If we want to check for type equality at $ { \color{black}{C} } $, we need to make sure that the types themselves are checked at $ { \color{black}{C} } $. However, checking types at $ { \color{black}{C} } $ would rule out variables marked at $ { \color{black}{\top} } $ from appearing in them. This would restrict us from expressing many examples, including the polymorphic identity function.

To move out of this impasse, we take inspiration from EPTS~\citep{mishra,mishra-linger:phd}. The key idea, adapted from \citet{mishra}, is to use a judgment of the form $  { \color{black}{C} }   \wedge  \Omega  \vdash  \ottnt{a} :^{  { \color{black}{C} }  }  \ottnt{A} $ instead of a judgment of the form $ \Omega  \vdash  \ottnt{a} :^{  { \color{black}{\top} }  }  \ottnt{A} $. The operation $ { \color{black}{C} }   \wedge  \Omega$ takes the point-wise meet of the labels in the context $\Omega$ with $ { \color{black}{C} } $, essentially reducing any label marked as $ { \color{black}{\top} } $ to $ { \color{black}{C} } $, making it available for use in a $ { \color{black}{C} } $-expression. This operation, called \emph{truncation}, makes $ { \color{black}{\top} } $ marked variables available at $ { \color{black}{C} } $. Other systems also use similar mechanisms for tracking irrelevance --- for example, we can see a relation between this idea and analogous ones in \cite{pfenning:2001} and \cite{abel}. In these systems, ``context resurrection'' operation makes proof variables and irrelevant variables in the context available for use, similar to
how $C \wedge \Omega$ makes $\top$-marked variables in the context available for use. 

\subsection{DDC: Basics}

Next, we design a general dependency analyzing calculus, \DDC{}, 
that takes advantage of compile-time irrelevance in its type system.
\DDC{} is a generalization of $\textsc{DDC}^{\top}$ and
$\text{EPTS}^{\bullet}$ \citep{mishra}. When $ { \color{black}{C} } $ equals $ { \color{black}{\top} } $, DDC
degenerates to $\textsc{DDC}^{\top}$, that does not use compile-time
irrelevance. When $ { \color{black}{C} } $ equals $ { \color{black}{\bot} } $, 
DDC degenerates to $\text{EPTS}^{\bullet}$, that identifies compile-time and run-time
irrelevance. A crucial distinction between $\text{EPTS}^{\bullet}$ and DDC is that
while the former is tied to a two element lattice, the latter can use any lattice. 
Thus, not only can \DDC{} distinguish between run-time and compile-time irrelevance,
but also it can simultaneously track other dependencies.

\begin{figure}[h]
  \drules[T]
  {$ \Omega  \vdash  \ottnt{a} :^{ \ell }  \ottnt{A} $}{\DDC{} core typing rules}{Var,Type,Pi,AbsC,AppC,ConvC} 
\begin{drulepar*}[]{$ \Omega \Vdash   \ottnt{a} :^{ \ell }  \ottnt{A} $}{Truncate at $ { \color{black}{\top} } $}
    \drule{CT-Leq}
    \drule[width=3in]{CT-Top}
  \end{drulepar*}
\caption{Dependent type system with compile-time irrelevance (core rules)}
\label{fig:ddc}
\label{fig:truncation}
\end{figure}

The core typing rules of DDC appear in Figure~\ref{fig:ddc}. Compared to
$\textsc{DDC}^{\top}$, this type system maintains the invariant that for any
$ \Omega  \vdash  \ottnt{a} :^{ \ell }  \ottnt{A} $, we have $\ell  \leq   { \color{black}{C} } $. To ensure that this is
the case, \rref{T-Type} and \rref{T-Var} include this precondition. This restriction means that we cannot 
really derive any term at $ { \color{black}{\top} } $ in DDC. We can get around this restriction by deriving $  { \color{black}{C} }   \wedge  \Omega  \vdash  \ottnt{a} :^{  { \color{black}{C} }  }  \ottnt{A} $ in place of $ \Omega  \vdash  \ottnt{a} :^{  { \color{black}{\top} }  }  \ottnt{A} $. 

Wherever $\textsc{DDC}^{\top}$ uses $ { \color{black}{\top} } $ as the observer level on a typing judgment, DDC uses truncation and level $ { \color{black}{C} } $ instead. If $\textsc{DDC}^{\top}$ uses some grade other than $ { \color{black}{\top} } $ as the observer level, DDC leaves the derivation as such. So a $\textsc{DDC}^{\top}$ judgment $ \Omega  \vdash  \ottnt{a} :^{ \ell }  \ottnt{A} $ is replaced with a \emph{truncated-at-top judgment}, $ \Omega \Vdash   \ottnt{a} :^{ \ell }  \ottnt{A} $ which can be read as: if $\ell =  { \color{black}{\top} } $, use the truncated version $  { \color{black}{C} }   \wedge  \Omega  \vdash  \ottnt{a} :^{  { \color{black}{C} }  }  \ottnt{A} $; otherwise use the normal version $ \Omega  \vdash  \ottnt{a} :^{ \ell }  \ottnt{A} $, as we see in Figure \ref{fig:truncation}. In the typing rules, uses of this new judgment have been highlighted in gray to emphasize the modification with respect to \DDCt{}.

\subsection{$\Pi$-types}

\Rref{T-Pi} is unchanged. The lambda \rref{T-AbsC} now checks the type at $ { \color{black}{C} } $ after truncating the variables in the context to $ { \color{black}{C} } $. The application \rref{T-AppC} checks the argument using the truncated-at-top judgment. Note that if $\ell_{{\mathrm{0}}} =  { \color{black}{\top} } $, the term $\ottnt{a}$ can depend upon any variable in $\Omega$. Such a dependence is allowed since information can always flow from relevant to  irrelevant contexts. 

To see how irrelevance works in this system, let's consider the
definition and use of the polymorphic identity function.

\begin{lstlisting}
   id : Π x:$^{ { \color{black}{\top} } }$Type. x -> x 
   id = λ$^{ { \color{black}{\top} } }$x. λ y. y
\end{lstlisting}

In \DDCt{}, the type \lstinline{Π x:$^{ { \color{black}{\top} } }$Type. x -> x} is checked at $ { \color{black}{\top} } $.
However, here it must be checked at level $ { \color{black}{C} } $, which requires the premise \hspace{2pt}\lstinline{x:$^{ { \color{black}{C} } }$Type $\vdash$ x -> x :$^{ { \color{black}{C} } }$ Type}. Note that if we used the same grade for the bound variable $\ottmv{x}$ in \rref{T-Pi} and \rref{T-AbsC}, we would have been in trouble because variable \cd{x} is compile-time relevant while we check the type, even though it is irrelevant in the term.\footnote{This is why we fuse the graded modality with the dependent types. If they 
were separated, and we had to bind here, it would be a problem since a dependent function and its type have different restrictions vis-\`{a}-vis the bound variable. }


Finally, observe that \rref{T-ConvC} uses the definitional equality at
$ { \color{black}{C} } $ instead of $ { \color{black}{\top} } $ and that the new type is checked after truncation.



\subsection{$\Sigma$-types}
%
%


\[ \drule{T-WPairC} \qquad \drule{T-LetPairC} \]

We also need to modify the typing rules for $\Sigma$ types accordingly. In particular, when we create a pair, we check the first component using the truncated-at-top judgment. This is akin to how we check the argument in \rref{T-AppC}. Note that if $\ell_{{\mathrm{0}}} =  { \color{black}{\top} } $, the first component $\ottnt{a}$ is compile-time irrelevant. In such a situation, we cannot type-check the second projection since it requires the first projection, as we see in the derived\footnote{\texttt{strong\_exists.v:T\_wproj1,T\_wproj2}} projection rules below. So pairs having type $ \Sigma  \ottmv{x} \!\!:^{  { \color{black}{\top} }  }\!\! \ottnt{A} . \ottnt{B} $ can only be eliminated via pattern matching if $B$ mentions $x$. However, pairs having type $ \Sigma  \ottmv{x} \!\!:^{  { \color{black}{C} }  }\!\! \ottnt{A} . \ottnt{B} $ can be eliminated via projections.

For example, for an output of the \cd{filter} function, \cd{ys : $\Sigma$m:$\!^{ { \color{black}{C} } }$Nat. Vec m Bool}, we have \cd{$\pi_1$ ys} $:^{ { \color{black}{C} } }$ \cd{Nat} and \cd{$\pi_2$ ys : Vec ($\pi_1$ ys) Bool}. Note that (\cd{$\pi_1$ ys}) is visible at $ { \color{black}{C} } $ and is used in the type of (\cd{$\pi_2$ ys}). We can substitute (\cd{$\pi_1$ ys}) for \cd{m} in (\cd{Vec m Bool}) because \cd{m :$^{ { \color{black}{C} } }$Nat $\ \vdash$ Vec m Bool $\ :^{ { \color{black}{C} } }$ Type}. However, (\cd{$\pi_1$ ys}) cannot be used at $ { \color{black}{\bot} } $, so it will be erasable then.

\[ \drule[width=3in]{T-ProjOne} \qquad \drule[width=3in]{T-ProjTwoC} \]



%

\subsection{Non-interference}
\label{sec:ddc-geq}

DDC satisfies an analogous noninterference theorem to the one presented for
SDC, using suitable definitions for the \emph{grading} relation, written $ \Phi  \vdash  \ottnt{a}  :  \ell $, and \emph{indexed
indistiguishability}, written $ \Phi  \vdash  \ottnt{b_{{\mathrm{1}}}}  \sim_{ \ell }  \ottnt{b_{{\mathrm{2}}}} $. The complete definition of these judgements appears in \auxiliarymaterial.

\begin{lemma}[Typing implies grading\footnote{\texttt{typing.v:Typing\_Grade}}] \label{DDC:typing_grading}
If $ \Omega  \vdash\,  \ottnt{a}  \, :^{ \ell }  \,  \ottnt{A} $ then $  |  \Omega  |   \vdash  \ottnt{a}  :  \ell $.
\end{lemma}

\begin{lemma}[Equivalence\footnote{\texttt{geq.v:GEq\_refl,GEq\_symmetry,GEq\_trans}}] \label{DDC:ind_equivalence}
Indexed indistinguishability at $\ell$ is an equivalence relation on well-graded terms at $\ell$.
\end{lemma}

\begin{lemma}[Indistinguishability under substitution\footnote{\texttt{subst.v:CEq\_GEq\_equality\_substitution}}] 
  If $  \Phi  ,   \ottmv{x}  \! :  \ell    \vdash  \ottnt{b_{{\mathrm{1}}}}  \sim_{ k }  \ottnt{b_{{\mathrm{2}}}} $ and $ \Phi  \vdash^{ \ell }_{ k }  \ottnt{a_{{\mathrm{1}}}}  \sim  \ottnt{a_{{\mathrm{2}}}} $ then $ \Phi  \vdash  \ottnt{b_{{\mathrm{1}}}}  \ottsym{\{}  \ottnt{a_{{\mathrm{1}}}}  \ottsym{/}  \ottmv{x}  \ottsym{\}}  \sim_{ k }  \ottnt{b_{{\mathrm{2}}}}  \ottsym{\{}  \ottnt{a_{{\mathrm{2}}}}  \ottsym{/}  \ottmv{x}  \ottsym{\}} $.
\end{lemma}

\begin{theorem}[Non-interference for DDC\footnote{\texttt{geq.v:CEq\_GEq\_respects\_Step}}]
\label{lemma:DDC:GEq_respects_Step}
  If $ \Phi  \vdash  \ottnt{a_{{\mathrm{1}}}}  \sim_{ k }  \ottnt{a'_{{\mathrm{1}}}} $ and $ \ottnt{a_{{\mathrm{1}}}}  \leadsto  \ottnt{a_{{\mathrm{2}}}} $ then there exists some $\ottnt{a'_{{\mathrm{2}}}}$ such that $ \ottnt{a'_{{\mathrm{1}}}}  \leadsto  \ottnt{a'_{{\mathrm{2}}}} $ and $ \Phi  \vdash  \ottnt{a_{{\mathrm{2}}}}  \sim_{ k }  \ottnt{a'_{{\mathrm{2}}}} $.
\end{theorem}

\subsection{Consistency of Equality}

\label{consisteq}

The equality relation of DDC incorporates compile-time irrelevance. To show that the
type system is sound, we need to show that the equality relation is consistent. Consistency of definitional equality means that there is no derivation that equates two
types having different head forms. For example, it should not equate $ \mathbf{Nat} $ with $\ottkw{Unit}$.

Note that if $ { \color{black}{\top} } $ inputs can interfere with $ { \color{black}{C} } $ outputs, the equality relation cannot
be consistent. To see why, let $  \ottmv{x} \! :^{  { \color{black}{\top} }  }\! \ottnt{A}   \vdash \,  \ottnt{b}  \, :^{  { \color{black}{C} }  } \,  \ottkw{Bool} $ and for $\ottnt{a_{{\mathrm{1}}}} , \ottnt{a_{{\mathrm{2}}}} : A$, let 
the terms $\ottnt{b}  \ottsym{\{}  \ottnt{a_{{\mathrm{1}}}}  \ottsym{/}  \ottmv{x}  \ottsym{\}}$ and $\ottnt{b}  \ottsym{\{}  \ottnt{a_{{\mathrm{2}}}}  \ottsym{/}  \ottmv{x}  \ottsym{\}}$ reduce to $\ottkw{True}$ and $\ottkw{False}$ 
respectively. Now, $  \ottsym{(}   \lambda^{  { \color{black}{\top} }  }  \ottmv{x} . \ottkw{if} \, \ottnt{b} \, \ottkw{then} \,  \mathbf{Nat}  \, \ottkw{else} \, \ottkw{Unit}   \ottsym{)}  \;  \ottnt{a_{{\mathrm{1}}}} ^{  { \color{black}{\top} }  }   \equiv_{  { \color{black}{C} }  }   \ottsym{(}   \lambda^{  { \color{black}{\top} }  }  \ottmv{x} . \ottkw{if} \, \ottnt{b} \, \ottkw{then} \,  \mathbf{Nat}  \, \ottkw{else} \, \ottkw{Unit}   \ottsym{)}  \;  \ottnt{a_{{\mathrm{2}}}} ^{  { \color{black}{\top} }  }  $. But then, by $\beta$-equivalence $  \mathbf{Nat}   \equiv_{  { \color{black}{C} }  }  \ottkw{Unit} $.

To prove consistency, we construct a standard parallel
reduction relation and show that this relation is confluent. Thereafter, we prove 
that if two terms are definitionally equal at $\ell$, then they are joinable at $\ell$, meaning they reduce, through parallel reduction, to two terms that are indistinguishable at $\ell$. Next, we show that joinability at $\ell$ implies consistency. Therefore, we conclude that for any $\ell$, the equality relation at $\ell$ is consistent. This implies that the equality relation at $ { \color{black}{C} } $, that ignores sub-terms marked with $ { \color{black}{\top} } $, is sound. Hence, DDC tracks compile-time irrelevance correctly. Note that DDC can track run-time irrelevance the same way as $\textsc{DDC}^{ { \color{black}{\top} } }$. 

We formally state consistency in terms of \emph{head forms}, i.e. syntactic
forms that correspond to types such as sorts $\ottnt{s}$, $\ottkw{Unit}$, $ \Pi  \ottmv{x} \!:^{ \ell }\! \ottnt{A} . \ottnt{B} $, etc.
\begin{theorem}[Consistency\footnote{\texttt{consist.v:DefEq\_Consistent}}]
If $ \Phi  \vdash  \ottnt{a}  \equiv_{ \ell }  \ottnt{b} $, and $\ottnt{a}$ and $\ottnt{b}$ both are head forms, then they have the same head form.
\end{theorem}



\subsection{Soundness theorem}

DDC is type sound and we have checked this and other results using the
Coq proof assistant. Below, we give an overview of the important lemmas in this
development.

The properties below are stated for \DDC{}, but they also apply to $\textsc{DDC}^{\top}$ since \DDC{} degenerates to \DDCt{} whenever $ { \color{black}{C} }  =  { \color{black}{\top} } $.
First, we list the properties related to grading that hold for all judgments:
indexed indistinguishability, definitional equality, and typing. (We only state
the lemmas for typing, their counterparts are analogous.) These lemmas are similar to
their simply-typed counterparts in Section~\ref{sec:sdc-metatheory}.

\begin{lemma}[Narrowing\footnote{\texttt{narrowing.v:Typing\_narrowing}}]
If $ \Omega  \vdash  \ottnt{a} :^{ \ell }  \ottnt{A} $ and $\Omega'  \leq  \Omega$, then $ \Omega'  \vdash  \ottnt{a} :^{ \ell }  \ottnt{A} $
\end{lemma}

\begin{lemma}[Weakening\footnote{\texttt{weakening.v:Typing\_weakening}}]
If $  \Omega_{{\mathrm{1}}} ,  \Omega_{{\mathrm{2}}}   \vdash  \ottnt{a} :^{ \ell }  \ottnt{A} $ then $   \Omega_{{\mathrm{1}}} ,  \Omega  ,  \Omega_{{\mathrm{2}}}   \vdash  \ottnt{a} :^{ \ell }  \ottnt{A} $.
\end{lemma}

\begin{lemma}[Restricted upgrading\footnote{\texttt{pumping.v:Typing\_pumping}}]
If $   \Omega_{{\mathrm{1}}} ,   \ottmv{x} \! :^{ \ell_{{\mathrm{0}}} }\! \ottnt{A}   ,  \Omega_{{\mathrm{2}}}   \vdash  \ottnt{b} :^{ \ell }  \ottnt{B} $ and $\ell_{{\mathrm{1}}}  \leq  \ell$ then 
   $    \Omega_{{\mathrm{1}}} ,   \ottmv{x} \! :^{   \ell_{{\mathrm{0}}}  \vee  \ell_{{\mathrm{1}}}   }\! \ottnt{A}    ,  \Omega_{{\mathrm{2}}}   \vdash  \ottnt{b} :^{ \ell }  \ottnt{B} $.
\end{lemma}

Next, we list some properties that are specific to the typing judgment. For any typing judgment in DDC,
the observer grade $\ell$ is at most $ { \color{black}{C} } $. Further, the observer grade of any judgment can be
raised up to $ { \color{black}{C} } $.

\begin{lemma}[Bounded by $ { \color{black}{C} } $\footnote{\texttt{pumping.v:Typing\_leq\_C}}]
If $ \Omega  \vdash  \ottnt{a} :^{ \ell }  \ottnt{A} $ then $\ell  \leq   { \color{black}{C} } $.
\end{lemma}

\begin{lemma}[Subsumption\footnote{\texttt{typing.v:Typing\_subsumption}}]
If $ \Omega  \vdash  \ottnt{a} :^{ \ell }  \ottnt{A} $ and $\ell  \leq  k$ and $k  \leq   { \color{black}{C} } $ then $ \Omega  \vdash  \ottnt{a} :^{ k }  \ottnt{A} $
\end{lemma}
%

Note that we don't require contexts to be well-formed in the typing judgment; we add context 
well-formedness constraints, as required, to our lemmas. The following lemmas are true for well-formed contexts.  A context $\Omega$ is well-formed, expressed as $\vdash  \Omega$, iff for any assumption
$\ottmv{x}:^\ell\ottnt{A}$ in $\Omega$, we have $ \Omega' \Vdash   \ottnt{A} :^{  { \color{black}{\top} }  }   \ottnt{s}  $, where
$\Omega'$ is the prefix of $\Omega$ that appears before the assumption.

\begin{lemma}[Substitution\footnote{\texttt{typing.v:Typing\_substitution\_CTyping}}]
If $   \Omega_{{\mathrm{1}}} ,   \ottmv{x} \! :^{ \ell_{{\mathrm{0}}} }\! \ottnt{A}   ,  \Omega_{{\mathrm{2}}}   \vdash  \ottnt{b} :^{ \ell }  \ottnt{B} $ and $\vdash  \Omega_{{\mathrm{1}}}$ and $ \Omega_{{\mathrm{1}}} \Vdash   \ottnt{a} :^{ \ell_{{\mathrm{0}}} }  \ottnt{A} $
then $  \Omega_{{\mathrm{1}}} ,  \Omega_{{\mathrm{2}}}   \ottsym{\{}  \ottnt{a}  \ottsym{/}  \ottmv{x}  \ottsym{\}}  \vdash  \ottnt{b}  \ottsym{\{}  \ottnt{a}  \ottsym{/}  \ottmv{x}  \ottsym{\}} :^{ \ell }  \ottnt{B}  \ottsym{\{}  \ottnt{a}  \ottsym{/}  \ottmv{x}  \ottsym{\}} $
\end{lemma}

Next, if a term is well-typed in our system, the type itself is also well-typed. 
\begin{lemma}[Regularity\footnote{\texttt{typing.v:Typing\_regularity}}]
  If $ \Omega  \vdash  \ottnt{a} :^{ \ell }  \ottnt{A} $ and $\vdash  \Omega$ then $ \Omega \Vdash   \ottnt{A} :^{  { \color{black}{\top} }  }   \ottnt{s}  $.
\end{lemma}

\noindent Finally, we have the two main lemmas proving type soundness.

\begin{lemma}[Preservation\footnote{\texttt{typing.v:Typing\_preservation}}]
If $ \Omega  \vdash  \ottnt{a} :^{ \ell }  \ottnt{A} $ and $\vdash  \Omega$ and  $ \ottnt{a}  \leadsto  \ottnt{a'} $, then $ \Omega  \vdash  \ottnt{a'} :^{ \ell }  \ottnt{A} $.
\end{lemma}

\begin{lemma}[Progress\footnote{\texttt{progress.v:Typing\_progress}}]
If $ \varnothing  \vdash  \ottnt{a} :^{ \ell }  \ottnt{A} $ then either $\ottnt{a}$ is a value or there exists some
$\ottnt{a'}$ such that $ \ottnt{a}  \leadsto  \ottnt{a'} $.
\end{lemma}

Hence, DDC is type sound. We have seen earlier that it tracks run-time and compile-time irrelevance correctly.

DDC is parameterized by a generic pure type system and a generic lattice. When the parameterizing pure type system is strongly normalizing, such as the Calculus of Constructions, type-checking is decidable. In the next section, we provide a demonstration.

\section{Type Checking}
\label{sec:typechecking}

\newcommand{\IS}{$\text{ICC}^{\ast}$}

As a pure type system, not all instances of DDC admit decidable type
checking. For example, in the presence of the \textsf{type}:\textsf{type}
axiom, the system includes non-terminating computations via Girard's
paradox. As as a result, we cannot decide equality in that system, so type
checking will be undecidable.  However, if the sorts, axioms and rules are
chosen such that the language is strongly normalizing, then we can define a
decidable type checking algorithm. This algorithm is standard, but relies on a
decision procedure for the equality judgement.

Our consistency proof, described in Section~\ref{consisteq}, gives us a
start. This proof uses an auxiliary binary relation called \emph{joinability},
which holds when two terms can use multiple steps of parallel reduction to
reach two simpler terms that differ only in their unobservable components. Joinability
and definitional equality induce the same relation on DDC terms. We can show
that two DDC terms are definitionally equal if and only if they
are joinable\footnote{\texttt{consist.v:DefEq\_Joins,Joins\_DefEq}}, which
means that a decision procedure based on joinability will be sound and
complete for DDC's labeled definition of equivalence.

Therefore, the decidability of type checking reduces to showing strong
normalization. If we select the sorts, axioms and rules of DDC to match those
of the Calculus of Constructions~\citep{pts},
we believe that this result holds, but leave a direct proof for future
work. However, by translating this instance of DDC to \IS{}, we can show that
a sublanguage of this instance is strongly
normalizing. \IS{}~\citep{barras:icc-star}, is a version of the Implicit
Calculus of Constructions with annotations that support decidable type
checking, but because it includes only (relevant and irrelevant) $\Pi$-types, so we must
restrict our attention to the corresponding fragment of DDC.

We define the following translation, written $\widetilde{\cdot}$,
that converts DDC terms to \IS{} terms.  The key parts of this translation map
arguments labeled $ { \color{black}{C} } $ and below to relevant arguments, and those
labeled greater than $ { \color{black}{C} } $, such as $ { \color{black}{\top} } $, to irrelevant
arguments.\footnote{The syntax of \IS{} uses parentheses to indicate usual
  (relevant) arguments and square brackets to indicate arguments that are
  irrelevant at both run time and compile time.}

\begin{align*}
 \widetilde{ \ottmv{x} }  = \ottmv{x} & \hspace*{20pt}  \widetilde{   \ottnt{s}   }  =  \ottnt{s}  & 
 \widetilde{  \Pi  \ottmv{x} \!:^{ \ell }\! \ottnt{A} . \ottnt{B}  }  & = \begin{cases}
                                 \Pi ( \ottmv{x}  \!:\!  \widetilde{ \ottnt{A} }  ).  \widetilde{ \ottnt{B} }   \;\; \text{if}\: \ell  \leq   { \color{black}{C} }  \\
                                 \Pi [  \ottmv{x}  \!:\!  \widetilde{ \ottnt{A} }   ] .  \widetilde{ \ottnt{B} }   \;\; \text{otherwise}
                              \end{cases} \\
 \widetilde{  \lambda  \ottmv{x} \!:^{ \ell }\! \ottnt{A}  .  \ottnt{b}  }  & = \begin{cases}
                                 \lambda (  \ottmv{x} \!:\!  \widetilde{ \ottnt{A} }   ) .   \widetilde{ \ottnt{b} }   \;\; \text{if}\: \ell  \leq   { \color{black}{C} }  \\
                                 \lambda [  \ottmv{x} \!:\!  \widetilde{ \ottnt{A} }   ] .   \widetilde{ \ottnt{b} }   \;\; \text{otherwise}
                              \end{cases} &
 \widetilde{   \ottnt{b}  \;  \ottnt{a} ^{ \ell }   }  & =   \begin{cases}
                         \widetilde{ \ottnt{b} }   \; (   \widetilde{ \ottnt{a} }   )  \;\; \text{if}\: \ell  \leq   { \color{black}{C} }  \\
                         \widetilde{ \ottnt{b} }   \; [   \widetilde{ \ottnt{a} }   ]  \;\; \text{otherwise} 
                       \end{cases}                                                
\end{align*}

Note that \IS{} compares terms for equality after an erasure 
operation, written $\cdot^{\ast}$, that removes all irrelevant arguments.
Now, we can show that the above translation preserves definitional equality and typing. 
Here, $ \widetilde{ \Omega } $ denotes $\Omega$ with the labels at
the variable bindings omitted.  
\begin{lemma}[Translation preservation] \label{DDCToICCStar}
If $ \Phi  \vdash  \ottnt{A}  \equiv_{  { \color{black}{C} }  }  \ottnt{B} $, then $   \widetilde{ \ottnt{A} }  ^{\ast}   \cong_{\beta\eta}    \widetilde{ \ottnt{B} }  ^{\ast}  $. \\
If $ \Omega  \vdash \,  \ottnt{a}  \, :^{ \ell } \,  \ottnt{A} $, then $  \widetilde{ \Omega }   \vdash\,   \widetilde{ \ottnt{a} }   \, :  \,   \widetilde{ \ottnt{A} }  $. 
\end{lemma}
Next, observe that because $\beta$-reductions are preserved by the
translation, any parallel reduction in DDC between terms $\ottnt{a}$ and $\ottnt{b}$ at level
$ { \color{black}{C} } $, where $\ottnt{a} \neq \ottnt{b}$, would correspond to a sequence of reduction steps
$ \widetilde{ \ottnt{a} }  \rightarrow^{+}_{\beta_{ie}}  \widetilde{ \ottnt{b} } $ in \IS{}. That means that an infinite sequence of parallel reductions $\ottnt{a_{{\mathrm{0}}}}$, $\ottnt{a_{{\mathrm{1}}}}$, \ldots, where
each term differs from the previous, corresponds to an infinite sequence of
reductions $ \widetilde{ \ottnt{a_{{\mathrm{0}}}} } $, $ \widetilde{ \ottnt{a_{{\mathrm{1}}}} } $ \ldots in \IS{}.  Therefore, as all
well-typed \IS{} terms are strongly normalizing, we can conclude that this
is so for this instance of DDC.

\paragraph{Non-terminating instances of DDC.}
For pure type systems that are not strongly normalizing, such as the
\textsf{type}:\textsf{type} language, there is an alternative approach to
developing a calculus with decidable type checking, following
\citet{weirich:systemd}. The key idea is to develop an annotated version of
DDC that book-keeps additional information from typing and equality derivations. 
In such an annotated version, the conversion rule would include an explicit coercion 
annotation that witnesses the equality between the concerned types, thus avoiding the need for
normalization.

\section{Discussions and Related Work} \label{related}

\subsection{Irrelevance in Dependent Type Theories}
\label{sec:irrelevance-related}

Overall, compile-time and run-time irrelevance is a well-studied topic in
the design of dependent type systems. In some systems, the focus is only
on support for run-time irrelevance: see~\cite{McBride:2016,atkey,brady:idris2,miquel,mishra,matus}.  In other systems, the focus is on compile-time irrelevance: see~\cite{pfenning:2001,abel}. Some systems support both, but require them to overlap, such as ~\cite{barras:icc-star,mishra-linger:phd,weirich:systemd,Nuyts18}. The system of \citet{Moon:2021} does not require them to overlap but their type system does not make use of compile-time irrelevance in the conversion rule.

To compare, system $\textsc{DDC}^{\top}$, presented here, can support run-time irrelevance only and is similar to the core language of \citet{matus}. However, note that $\textsc{DDC}^{\top}$ can track dependencies in general while the system in \citet{matus} tracks run-time irrelevance alone. \textsc{DDC}, on the other hand, is the only system that we are aware of that tracks run-time and compile-time irrelevance separately and makes use of the latter in the conversion rule. Further, \textsc{DDC} tracks these irrelevances in the presence of strong $\Sigma$-types with erasable first components, something which, to the best of our knowledge, no prior work has been able to.

Prior work has identified the difficulty in handling strong $\Sigma$-types with erasable first
components in a setting that tracks compile-time irrelevance. \citet{abel} point out that strong irrelevant $\Sigma$-types make their theory inconsistent. Similarly, EPTS$^{\bullet}$~\citep{mishra-linger:phd} cannot define the projections for pairs having such $\Sigma$-types. The reason behind this is that EPTS$^{\bullet}$ is hard-wired to work with a two-element lattice which identifies compile-time and run-time irrelevance. As such, projections from such pairs lead to type unsoundness. For example, considering the first components to be run-time irrelevant, the pairs $ ( \ottkw{Int}  ,  \ottkw{unit} ) $ and $ ( \ottkw{Bool}  ,  \ottkw{unit} ) $ are run-time equivalent. Since EPTS$^{\bullet}$ identifies run-time and compile-time irrelevance, these pairs are also compile-time equivalent. Then, taking the first projections of these pairs, one ends up with $\ottkw{Int}$ and $\ottkw{Bool}$ being compile-time equivalent. We resolve this problem by distinguishing between run-time and compile-time irrelevance, thus requiring a lattice with  three elements.

Next, we compare our work with existing literature with respect to the equality relation. We analyze compile-time irrelevance to enable the equality relation to ignore unnecessary sub-terms. However, since our equality relation is untyped, we cannot include type-dependent
rules in our system, such as $\eta$-equivalence for the \cd{Unit} type.
Several prior works on irrelevance~\citep{miquel,barras:icc-star,mishra-linger:phd,matus} use an untyped equality relation.  However, some prior work, such as~\cite{pfenning:2001,abel}, do consider compile-time irrelevance in the context of typed-directed equality. But such systems require irrelevant arguments to functions appear only irrelevantly in the codomain type of the function, thus ruling out several examples including the polymorphic identity function.

\subsection{Quantitative Type Systems}

Our work is closely related to quantitative type systems \citep{petricek,Ghica:2014,Brunel:2014,McBride:2016,atkey,orchard,abel20,grad,Moon:2021}. Such systems provide a fine-grained accounting of coeffects, viewed as resources, for example, variable usage, linearity, liveness, etc. A typical judgment from a quantitative type system \citep{grad} may look like: \[    \ottmv{x} \! :^{  1  }\! \ottkw{Bool}  ,   \ottmv{y} \! :^{  1  }\! \ottkw{Int}   ,   \ottmv{z} \! :^{  0  }\! \ottkw{Bool}   \vdash \ottkw{if}\, x \, \ottkw{then}\, y + 1\, \ottkw{else}\, y - 1\, :^1\, \ottkw{Int} \] The variable $\ottmv{x}$ is used once in the condition, the variable $\ottmv{y}$ is used once in each of the branches while the variable $\ottmv{z}$ is not used at all. As such, they are marked with these quantities in the context. 

This form of judgment is very similar to our typing judgments with quantities appearing in place of levels. However, there is a crucial difference: to the right of the turnstile, while any level may appear in our judgments, only the quantity $ 1 $ can appear in typing judgments of quantitative systems. A quantitative system that allows an arbitrary quantity to the right of the turnstile is not closed under
substitution \citep{McBride:2016,atkey}. As such, quantitative systems are tied to a fixed reference while our systems can view programs from different reference levels. This difference in form results from the difference in the purposes the two kinds of systems serve: quantitative systems count while our systems compare. Counting requires a fixed standard or reference whereas comparison does not. Applications that require counting, like linearity tracking, are handled well by quantitative systems while applications that require comparison, like ensuring secure information flow, are handled well by systems of our kind.

From a type-theoretic standpoint, in general, quantitative systems cannot eliminate pairs through projections. This is so because there is no general way to split the resources of the context that type-checks a pair. Eliminating pairs through projections is straightforward in our systems because the grade on the typing judgment can control where the projections are visible.  

\subsection{Dependency Analysis and Dependent Type Theory}

Dependency analysis and dependent type theories have come together in some existing work.

Like our system, \cite{prost:lambda-cube} extends the $\lambda$-cube so that it may track
dependencies. However, unlike our system, this work uses sorts to track dependencies.
It is inspired by the distinction between sorts in the Calculus of Constructions
where computationally relevant and irrelevant terms live in sorts \cd{Set} and \cd{Prop} respectively.
As \citet{mishra-linger:phd} points out, such an approach ties up two distinct language features,
sorts and dependency analysis, which can be treated in a more orthogonal manner.

\citet{color} is very related to our work. They use colors to erase terms while we use grades. Colors and grades both form a lattice structure and their usage in the respective type systems are quite similar. However, \citet{color} use internalized parametricity to reason about erasure; so it is important that their type theory is logically consistent. Our work does not rely on the normalizing nature of the theory; we take a direct route to analyzing erasure.

Like our work, \citet{caires} track information flow in a dependent type system. But \citet{caires} focus on more imperative features, like modeling of state while we focus on irrelevance. A distinguishing feature of their system is that they allow security labels to depend upon terms, something that we don't attempt here.

\section{Conclusion} \label{conclusion}

We started with the aim of designing a dependent calculus that can analyze dependencies in general, and run-time and compile-time irrelevance in particular. Towards this end, we designed a simple dependency calculus, SDC, and then extended it to two dependent calculi, $\textsc{DDC}^{\top}$ and DDC. $\textsc{DDC}^{\top}$ can track run-time irrelevance while DDC can track both run-time and compile-time irrelevance along with other dependencies.

In future, we would like to explore how irrelevance interacts with other dependencies. We also want to explore whether our systems can be integrated with existing graded type systems, especially quantitative type systems. Yet another interesting direction for research is that how they compare with graded effect systems.

Our work lies in the intersection of dependency analysis and irrelevance tracking in dependent type systems. Both these areas have rich literature of their own. We hope that the connections established in this paper will be mutually beneficial and help in the future exploration of dependencies and irrelevance in dependent type systems.

\section{Acknowledgments}
\label{sec:acknowledgments}
The first two authors were supported by the National Science Foundation under 
Grant Nos. 1703835 and 1521539. The second author was supported by the National Science Foundation under Grant No. 2104535.

\newpage
\bibliography{qtt,weirich,eades,pritam}

\begin{thebibliography}{34}
\providecommand{\natexlab}[1]{#1}
\providecommand{\url}[1]{\texttt{#1}}
\providecommand{\urlprefix}{URL }
\expandafter\ifx\csname urlstyle\endcsname\relax
  \providecommand{\doi}[1]{doi:\discretionary{}{}{}#1}\else
  \providecommand{\doi}{doi:\discretionary{}{}{}\begingroup
  \urlstyle{rm}\Url}\fi

\bibitem[{Abadi et~al.(1999)Abadi, Banerjee, Heintze, and Riecke}]{dcc}
Abadi, M., Banerjee, A., Heintze, N., Riecke, J.G.: A core calculus of
  dependency. In: Proceedings of the 26th ACM SIGPLAN-SIGACT Symposium on
  Principles of Programming Languages, p. 147–160, POPL '99, Association for
  Computing Machinery, New York, NY, USA (1999), ISBN 1581130953,
  \doi{10.1145/292540.292555}

\bibitem[{Abel and Bernardy(2020)}]{abel20}
Abel, A., Bernardy, J.P.: A unified view of modalities in type systems. Proc.
  ACM Program. Lang. \textbf{4}(ICFP) (Aug 2020), \doi{10.1145/3408972}

\bibitem[{Abel and Scherer(2012)}]{abel}
Abel, A., Scherer, G.: On irrelevance and algorithmic equality in predicative
  type theory. Logical Methods in Computer Science \textbf{8}(1) (mar 2012),
  \doi{10.2168/lmcs-8(1:29)2012}

\bibitem[{Atkey(2018)}]{atkey}
Atkey, R.: Syntax and semantics of quantitative type theory. In: Proceedings of
  the 33rd Annual ACM/IEEE Symposium on Logic in Computer Science, p. 56–65,
  LICS '18, Association for Computing Machinery, New York, NY, USA (2018), ISBN
  9781450355834, \doi{10.1145/3209108.3209189}

\bibitem[{Barendregt(1993)}]{pts}
Barendregt, H.P.: Lambda Calculi with Types, p. 117–309. Oxford University
  Press, Inc., USA (1993), ISBN 0198537611

\bibitem[{Barras and Bernardo(2008)}]{barras:icc-star}
Barras, B., Bernardo, B.: The implicit calculus of constructions as a
  programming language with dependent types. In: Amadio, R. (ed.) Foundations
  of Software Science and Computational Structures, pp. 365--379, FOSSACS 2008,
  Springer Berlin Heidelberg, Budapest, Hungary (2008)

\bibitem[{Bernardy and Guilhem(2013)}]{color}
Bernardy, J.P., Guilhem, M.: Type-theory in color. SIGPLAN Not. \textbf{48}(9),
  61–72 (Sep 2013), ISSN 0362-1340, \doi{10.1145/2544174.2500577}

\bibitem[{Brady(2021)}]{brady:idris2}
Brady, E.: {Idris 2: Quantitative Type Theory in Practice}. In: M{\o}ller, A.,
  Sridharan, M. (eds.) 35th European Conference on Object-Oriented Programming
  (ECOOP 2021), Leibniz International Proceedings in Informatics (LIPIcs), vol.
  194, pp. 9:1--9:26, Schloss Dagstuhl -- Leibniz-Zentrum f{\"u}r Informatik,
  Dagstuhl, Germany (2021), ISBN 978-3-95977-190-0, ISSN 1868-8969,
  \doi{10.4230/LIPIcs.ECOOP.2021.9}

\bibitem[{Brunel et~al.(2014)Brunel, Gaboardi, Mazza, and
  Zdancewic}]{Brunel:2014}
Brunel, A., Gaboardi, M., Mazza, D., Zdancewic, S.: A core quantitative
  coeffect calculus. In: Shao, Z. (ed.) Programming Languages and Systems, pp.
  351--370, Springer Berlin Heidelberg, Berlin, Heidelberg (2014)

\bibitem[{Choudhury et~al.(2021)Choudhury, Eades~III, Eisenberg, and
  Weirich}]{grad}
Choudhury, P., Eades~III, H., Eisenberg, R.A., Weirich, S.: A graded dependent
  type system with a usage-aware semantics. Proc. ACM Program. Lang.
  \textbf{5}(POPL) (Jan 2021), \doi{10.1145/3434331}

\bibitem[{Choudhury et~al.(2022)Choudhury, {Eades III}, and
  Weirich}]{esop22:artifact}
Choudhury, P., {Eades III}, H., Weirich, S.: {Artifact associated with "A
  Dependent Dependency Calculus"} (Jan 2022), \doi{10.5281/zenodo.5903726}

\bibitem[{Denning(1976)}]{denning1}
Denning, D.E.: A lattice model of secure information flow. Commun. ACM
  \textbf{19}(5), 236–243 (May 1976), ISSN 0001-0782,
  \doi{10.1145/360051.360056}

\bibitem[{Eisenberg et~al.(2021)Eisenberg, Duboc, Weirich, and
  Lee}]{eisenberg:existentials}
Eisenberg, R.A., Duboc, G., Weirich, S., Lee, D.: An existential crisis
  resolved: Type inference for first-class existential types. Proc. ACM
  Program. Lang. \textbf{5}(ICFP) (Aug 2021),
  \urlprefix\url{https://richarde.dev/papers/2021/exists/exists.pdf}

\bibitem[{Ghica and Smith(2014)}]{Ghica:2014}
Ghica, D.R., Smith, A.I.: Bounded linear types in a resource semiring. In:
  European Symposium on Programming Languages and Systems, pp. 331--350,
  Springer (2014)

\bibitem[{Heintze and Riecke(1998)}]{slam}
Heintze, N., Riecke, J.G.: The slam calculus: Programming with secrecy and
  integrity. In: Proceedings of the 25th ACM SIGPLAN-SIGACT Symposium on
  Principles of Programming Languages, p. 365–377, POPL '98, Association for
  Computing Machinery, New York, NY, USA (1998), ISBN 0897919793,
  \doi{10.1145/268946.268976}

\bibitem[{Louren\c{c}o and Caires(2015)}]{caires}
Louren\c{c}o, L., Caires, L.: Dependent information flow types. In: Proceedings
  of the 42nd Annual ACM SIGPLAN-SIGACT Symposium on Principles of Programming
  Languages, p. 317–328, POPL '15, Association for Computing Machinery, New
  York, NY, USA (2015), ISBN 9781450333009, \doi{10.1145/2676726.2676994}

\bibitem[{McBride(2016)}]{McBride:2016}
McBride, C.: I Got Plenty o' Nuttin', pp. 207--233. Springer International
  Publishing, Cham (2016)

\bibitem[{Miquel(2001)}]{miquel}
Miquel, A.: The implicit calculus of constructions extending pure type systems
  with an intersection type binder and subtyping. In: Abramsky, S. (ed.) Typed
  Lambda Calculi and Applications, pp. 344--359, Springer Berlin Heidelberg,
  Berlin, Heidelberg (2001), ISBN 978-3-540-45413-7

\bibitem[{Mishra-Linger and Sheard(2008)}]{mishra}
Mishra-Linger, N., Sheard, T.: Erasure and polymorphism in pure type systems.
  In: Proceedings of the Theory and Practice of Software, 11th International
  Conference on Foundations of Software Science and Computational Structures,
  p. 350–364, FOSSACS'08/ETAPS'08, Springer-Verlag, Berlin, Heidelberg
  (2008), ISBN 3540784977

\bibitem[{Mishra-Linger(2008)}]{mishra-linger:phd}
Mishra-Linger, R.N.: Irrelevance, Polymorphism, and Erasure in Type Theory.
  Ph.D. thesis, Portland State University, Department of Computer Science
  (2008), \doi{10.15760/etd.2669}

\bibitem[{Moggi(1991)}]{moggi}
Moggi, E.: Notions of computation and monads. Information and Computation
  \textbf{93}(1), 55--92 (1991), ISSN 0890-5401,
  \urlprefix\url{https://www.sciencedirect.com/science/article/pii/0890540191900524},
  selections from 1989 IEEE Symposium on Logic in Computer Science

\bibitem[{Moon et~al.(2021)Moon, Eades~III, and Orchard}]{Moon:2021}
Moon, B., Eades~III, H., Orchard, D.: Graded modal dependent type theory. In:
  Yoshida, N. (ed.) Programming Languages and Systems, pp. 462--490, Springer
  International Publishing, Cham (2021), ISBN 978-3-030-72019-3

\bibitem[{Nuyts and Devriese(2018)}]{Nuyts18}
Nuyts, A., Devriese, D.: Degrees of relatedness: {A} unified framework for
  parametricity, irrelevance, ad hoc polymorphism, intersections, unions and
  algebra in dependent type theory. In: Dawar, A., Gr{\"{a}}del, E. (eds.)
  Proceedings of the 33rd Annual {ACM/IEEE} Symposium on Logic in Computer
  Science, {LICS} 2018, Oxford, UK, July 09-12, 2018, pp. 779--788, {ACM}
  (2018), \doi{10.1145/3209108.3209119},
  \urlprefix\url{https://doi.org/10.1145/3209108}

\bibitem[{Orchard et~al.(2019)Orchard, Liepelt, and Eades~III}]{orchard}
Orchard, D., Liepelt, V.B., Eades~III, H.: Quantitative program reasoning with
  graded modal types. Proc. ACM Program. Lang. \textbf{3}(ICFP) (Jul 2019),
  \doi{10.1145/3341714}

\bibitem[{Petricek et~al.(2014)Petricek, Orchard, and Mycroft}]{petricek}
Petricek, T., Orchard, D., Mycroft, A.: Coeffects: A calculus of
  context-dependent computation. In: Proceedings of International Conference on
  Functional Programming, ICFP 2014 (2014)

\bibitem[{Pfenning(2001)}]{pfenning:2001}
Pfenning, F.: Intensionality, extensionality, and proof irrelevance in modal
  type theory. In: Proceedings of the 16th Annual IEEE Symposium on Logic in
  Computer Science, pp. 221--, LICS '01, IEEE Computer Society, Washington, DC,
  USA (2001), \urlprefix\url{http://dl.acm.org/citation.cfm?id=871816.871845}

\bibitem[{Prost(2000)}]{prost:lambda-cube}
Prost, F.: A static calculus of dependencies for the $\lambda$-cube. In:
  Proceedings Fifteenth Annual IEEE Symposium on Logic in Computer Science
  (Cat. No.99CB36332), pp. 267--276 (2000), \doi{10.1109/LICS.2000.855775}

\bibitem[{Shikuma and Igarashi(2006)}]{igarashi}
Shikuma, N., Igarashi, A.: Proving noninterference by a fully complete
  translation to the simply typed $\lambda$-calculus. In: Proceedings of the
  11th Asian Computing Science Conference on Advances in Computer Science:
  Secure Software and Related Issues, p. 301–315, ASIAN'06, Springer-Verlag,
  Berlin, Heidelberg (2006), ISBN 3540775048

\bibitem[{Smith and Volpano(1998)}]{smith-volpano}
Smith, G., Volpano, D.: Secure information flow in a multi-threaded imperative
  language. In: Proceedings of the 25th ACM SIGPLAN-SIGACT Symposium on
  Principles of Programming Languages, p. 355–364, POPL '98, Association for
  Computing Machinery, New York, NY, USA (1998), ISBN 0897919793,
  \doi{10.1145/268946.268975}

\bibitem[{Sulzmann et~al.(2007)Sulzmann, Chakravarty, Jones, and
  Donnelly}]{systemfc}
Sulzmann, M., Chakravarty, M.M.T., Jones, S.P., Donnelly, K.: System f with
  type equality coercions. In: Proceedings of the 2007 ACM SIGPLAN
  International Workshop on Types in Languages Design and Implementation, p.
  53–66, TLDI '07, Association for Computing Machinery, New York, NY, USA
  (2007), ISBN 159593393X, \doi{10.1145/1190315.1190324}

\bibitem[{Teji\v{s}\v{c}\'{a}k(2020)}]{matus}
Teji\v{s}\v{c}\'{a}k, M.: A dependently typed calculus with pattern matching
  and erasure inference. Proc. ACM Program. Lang. \textbf{4}(ICFP) (Aug 2020),
  \doi{10.1145/3408973}

\bibitem[{Thiemann(1997)}]{binding-time}
Thiemann, P.: A unified framework for binding-time analysis. In: Proceedings of
  the 7th International Joint Conference CAAP/FASE on Theory and Practice of
  Software Development, p. 742–756, TAPSOFT '97, Springer-Verlag, Berlin,
  Heidelberg (1997), ISBN 3540627812

\bibitem[{Tip(1995)}]{slicing}
Tip, F.: A survey of program slicing techniques. Journal of Programming
  Languages \textbf{3} (1995)

\bibitem[{Weirich et~al.(2017)Weirich, Voizard, de~Amorim, and
  Eisenberg}]{weirich:systemd}
Weirich, S., Voizard, A., de~Amorim, P.H.A., Eisenberg, R.A.: A specification
  for dependent types in {Haskell}. Proc. ACM Program. Lang. \textbf{1}(ICFP),
  31:1--31:29 (Aug 2017), ISSN 2475-1421, \doi{10.1145/3110275}

\end{thebibliography}

\vfill

\ifextended
\appendix

\section{System Specification for SDC }
\label{app:sdc-rules}

\[
\begin{array}{llcll}
\textit{labels} & \ell, k & ::= &  { \color{black}{\bot} }  \mid  { \color{black}{\top} }  \mid  k  \wedge  \ell  \mid  k  \vee  \ell  \mid \ldots \\ 
\textit{types} & \ottnt{A},\ottnt{B}   & ::=& \ottkw{Unit} \mid   \ottnt{A}  \to  \ottnt{B}   \mid  \ottnt{A}  \times  \ottnt{B}  \mid  \ottnt{A}  +  \ottnt{B}   \mid  T^{ \ell }\;  \ottnt{A} \\
\textit{terms} & \ottnt{a}, \ottnt{b}   & ::=& \ottmv{x} \mid  \lambda \ottmv{x} \!:\! \ottnt{A} . \ottnt{a}  \mid  \ottnt{a}  \;  \ottnt{b}  & \mbox{\it variables and functions} \\ 
                              && \mid & \ottkw{unit}  \mid  ( \ottnt{a}  ,  \ottnt{b} )  \mid  \pi_1\  \ottnt{a}  \mid  \pi_2\  \ottnt{a}  & \mbox{\it unit and products} \\
                              && \mid &  \ottkw{inj}_1\,  \ottnt{a}  \mid  \ottkw{inj}_2\,  \ottnt{a}  \mid  \ottkw{case} \,  \ottnt{a} \, \ottkw{of}\,  \ottnt{b_{{\mathrm{1}}}}  ;  \ottnt{b_{{\mathrm{2}}}}  & \mbox{\it sums} \\
                              && \mid &  \eta^{ \ell }\;  \ottnt{a}  \mid  \ottkw{bind} ^{ \ell } \,  \ottmv{x}  =  \ottnt{a}  \,  \ottkw{in}  \,  \ottnt{b}  & \mbox{\it graded modality}
\\
\textit{contexts} & \Omega & ::= & \varnothing \mid  \Omega ,   \ottmv{x} \! :^{ \ell }\! \ottnt{A}  \\
\\
\end{array}
\]

\subsection{Typing and Operational Semantics}

\drules[SDC]{$ \Omega  \vdash\,  \ottnt{a}  \, :^{ \ell }  \,  \ottnt{A} $}{Typing rules for SDC}{Var,Unit,Abs,App,Pair,ProjOne,ProjTwo,InjOne,InjTwo,Case,Return,Bind}

\drules[SDCStep]{$ \ottnt{a}  \leadsto  \ottnt{a'} $}{CBN small step operational semantics for SDC}{AppCong,Beta,ProjOneCong,ProjTwoCong,ProjOneBeta,ProjTwoBeta,CaseCong,CaseOneBeta,CaseTwoBeta,BindCong,BindBeta}

\subsection{Indexed Indistinguishability}

\drules[SGEq]{$ \Phi  \vdash  \ottnt{a}  \sim_{ \ell }  \ottnt{b} $}{Indexed Syntactic Equality}{Var,Unit,Abs,App,Return,Bind,Pair,ProjOne,ProjTwo,InjOne,InjTwo,Case}
\drules[SEq]{$ \Phi  \vdash^{ \ell_{{\mathrm{0}}} }_{ \ell }  \ottnt{a}  \sim  \ottnt{b} $}{Conditional Syntactic Equality}{Leq,Nleq}

\section{System Specification for DDC }
\label{app:ddc-rules}

This is the complete system that we have formalized using the Coq proof
assistant. The type setting of all of these rules is generated from the same Ott
source that also generates the Coq definitions. Some of the following rules may appear different from
their presentation in the paper. The reason behind this is that Ott is not very good at handling multiple
variable bindings. So, when necessary, we replace expressions involving multiple bound variables with equivalent expressions that have only single bound variable.   

\subsection{Operational semantics}

\drules[ValueType]{}{Values that are types}{Type,Pi,WSigma,Sum,Unit}
\drules[V]{}{Values}{ValueType,TmUnit,WPair,InjOne,InjTwo}

\drules[S]{$ \ottnt{a}  \leadsto  \ottnt{a'} $}{CBN small-step operational semantics}{AppCong,Beta,CaseCong,CaseOneBeta,CaseTwoBeta,
  LetPairCong,LetPairBeta}

\subsection{Definitional equality}

\drules[Eq]{$ \Phi  \vdash  \ottnt{a}  \equiv_{ \ell }  \ottnt{b} $}{Definitional Equality}{Refl,Sym,Trans,SubstIrrel,Beta,Pi,Abs,App,PiFst,PiSnd,
    WSigma,WSigmaFst,WSigmaSnd,WPair,LetPair,
    Sum,SumFst,SumSnd,InjOne,InjTwo,Case,TyUnit,TmUnit}
\drules[CEq]{$ \Phi  \vdash^{ k }  \ottnt{a}  \equiv_{ \ell }  \ottnt{b} $}{Conditional Definitional Equality}{Leq,Nleq}

\drules[G]{$ \Phi  \vdash  \ottnt{a}  :  k $}{Grading}{Type,Var,Pi,Abs,App,WSigma,WPair,LetPair,
                             Sum,InjOne,InjTwo,Case,TyUnit,TmUnit}
\drules[CG]{$ \Phi  \vdash_{ k }^{ \ell }  \ottnt{a} $}{Conditional Grading}{Leq,Nleq}

\subsection{Type System}

\drules[T]{$ \Omega  \vdash  \ottnt{a} :^{ \ell }  \ottnt{A} $}{Typing}{Type,Conv,Var,Pi,Abs,App,AppIrrel,WSigma,WPair,WPairIrrel,LetPairCA,
              Sum,InjOne,InjTwo,CaseC,TmUnit,TyUnit}

\subsection{Indexed Indistinguishability} \label{dep-indist}

\drules[GEq]{$ \Phi  \vdash  \ottnt{a}  \sim_{ \ell }  \ottnt{b} $}{Indexed Indistinguishability}{Type,Var,Pi,Abs,App,WSigma,WPair,LetPair,Sum,InjOne,InjTwo,Case,TyUnit,TmUnit}
\drules[CEq]{$ \Phi  \vdash^{ \ell_{{\mathrm{0}}} }_{ \ell }  \ottnt{a}  \sim  \ottnt{b} $}{Conditional Indistinguishability}{Leq,Nleq}

\subsection{Auxiliary Judgments}
\label{parconsist}

\drules[Par]{$ \Phi   \vdash    \ottnt{a}  \Rightarrow_{ \ell }  \ottnt{b} $}{Parallel reduction}{Refl,Pi,AppBeta,App,Abs,WSigma,WPair,WPairBeta,LetPair,
  Sum,InjOne,InjTwo,CaseBetaOne,CaseBetaTwo,Case}
\drules[CPar]{$ \Phi  \vdash^{ k }_{ \ell }  \ottnt{a}   \Rightarrow   \ottnt{b} $}{Conditional Parallel reduction}{Leq,Nleq}
\drules[MP]{$ \Phi   \vdash   \ottnt{a}   \Rightarrow^{\ast} _{ \ell }  \ottnt{b} $}{Parallel reduction, reflexive transitive closure}
   {Refl,Step}
\drules[]{$ \Phi \vdash   \ottnt{a}   \Leftrightarrow _{ \ell }   \ottnt{b} $}{Joinability}
   {join}
\drules[ConsistentXXa]{$ \mathsf{Consistent}\  \ottnt{a}  \  \ottnt{b} $}{Consistent Head Forms}{Type,Unit,Pi,WSigma,Sum,StepXXR,StepXXL}
\fi

\end{document}